\documentclass[sigconf,authorversion]{acmart}

\usepackage[utf8]{inputenc}
\usepackage[T1]{fontenc}
\usepackage[shortcuts]{extdash}
\usepackage{csquotes}
\usepackage{xargs}
\usepackage{xspace}
\usepackage{ifthen}
\usepackage[ruled,lined,linesnumbered]{algorithm2e}
\usepackage{amsmath,amssymb,amsthm} % for \DeclareMathOperator
\usepackage{bm,upgreek}
\usepackage[mathscr]{euscript}
\usepackage{dsfont}
\usepackage{xcolor}
\usepackage{graphicx}
\usepackage{subcaption}
\graphicspath{{./figures/}}

\usepackage[disable]{todonotes} %disable
\newcommand{\TODO}[2][inline]{\todo[#1]{\texttt{TODO:} #2}}

\usepackage{tabularx,booktabs}
\usepackage{mathtools}

\usepackage{aliascnt}
\usepackage{hyperref}
\hypersetup{
	colorlinks,
	linkcolor={black},
	citecolor={red!60!black},
	urlcolor={blue!60!black}
}
\usepackage[nameinlink]{cleveref}
\crefname{figure}{figure\xspace{}}{figures\xspace{}}
\Crefname{figure}{Figure\xspace{}}{Figures\xspace{}}

\usepackage{glossaries}
\newacronym[longplural={Bags-of-Words}]{BoW}{BoW}{Bag-of-Words}
\newacronym{tf}{tf}{term frequency}
\newacronym{tfidf}{tf--idf}{term frequency--inverse document frequency}
\newacronym{tfvec}{tf vector}{term frequency vector}
\newacronym{tfidfvec}{tf--idf vector}{term frequency--inverse document frequency vector}
\newacronym{PII}{PII}{personally identifiable information}
\newacronym{WP}{WritePrints}{WritePrints} %{WP}
\newacronym{WPL}{WritePrints (Limited)}{WritePrints (Limited)} %{WP-Limited}
\newacronym{VSM}{VSM}{vector space model}
\newacronym{NLP}{NLP}{natural language processing}
\newacronym{TM}{TM}{text mining}
\newacronym{IR}{IR}{information retrieval}

\providecommand{\numberset}[1]{\ensuremath{\mathbb{#1}}}
\providecommand{\Reals}{\numberset{R}} % reals
\providecommand{\RealsGEQ}[1]{\Reals_{\geq{#1}}} % reals
\providecommand{\Nats}{\numberset{N}} % natural numbers
\DeclarePairedDelimiter{\parens}{\lparen}{\rparen}
\DeclarePairedDelimiter{\set}{\lbrace}{\rbrace}

\providecommand{\interval}[2][0]{\ensuremath{\left[#1,#2\right]}} % interval from 0
\providecommand{\iinterval}[2][1]{\interval[#1]{#2}} % set range from 1
\providecommand{\irange}[2][1]{\iinterval[#1]{#2}} %{\range[#1]{#2}} % set range from 1

\DeclarePairedDelimiter{\abs}{\lvert}{\rvert}
\DeclarePairedDelimiter{\norm}{\lVert}{\rVert} %\newcommand{\norm}[1]{\left\lVert #1 \right\rVert}
\DeclareMathOperator{\dist}{d}

\newcommand{\dX}{\ensuremath{\dist_\X}}

\newcommand{\inv}[1][-1]{^{#1}}

\DeclarePairedDelimiter{\brackets}{\lbrack}{\rbrack}
\newcommand{\pr}[2][]{\Pr\brackets[#1]{#2}}

\DeclareMathOperator{\Cat}{Cat} % categorical distribution
\providecommand{\samplefrom}{\ensuremath{\leftarrow_\mathrm{R}}}

\newcommand{\script}[1]{\ensuremath{\mathcal{#1}}}
\providecommand{\X}{\script{X}}
\providecommand{\Y}{\script{Y}}
\providecommand{\Z}{\script{Z}}
\newcommand{\T}{\script{T}}
\newcommand{\Universe}{\script{V}}

\newcommand{\RandVars}{\script{R}} % random variables

\providecommand{\rating}{\rho}
\newcommand{\mechanism}[1]{\ensuremath{\mathcal{#1}}}
\newcommandx{\Em}[2][1=\epsilon,2=\rating]{\ensuremath{\mechanism{E}_{#1,#2}}}
\newcommandx{\Lm}[2][1=\epsilon,2=f]{\ensuremath{\mechanism{L}_{#1,#2}}}
\newcommand{\M}{\mechanism{M}}

\newcommand{\A}{\mechanism{A}}
\newcommand{\B}{\mechanism{B}}

\providecommand{\dlp}{differential privacy\xspace}

\providecommand{\dyp}{differentially private\xspace}

\newcommand{\adj}{\sim}

\newcommand{\doc}{\ensuremath{t}} % avoid same symbol for [d]ocument and [d]istance (metric)
\newcommand{\docvec}{\ensuremath{\mathbf{\doc}}}

\newcommand{\comp}{\ensuremath{\theta}}
\newcommand{\compvec}{\ensuremath{\bm{\uptheta}}}

\newcommand{\res}{\ensuremath{s}}
\newcommand{\resvec}{\ensuremath{\mathbf{s}}}
\newcommand{\reslen}{\ensuremath{n}}

\newcommand{\Vocab}{\ensuremath{V}}
\newcommand{\VocSize}{\ensuremath{L}}

\newcommandx{\smexp}[4][1=\epsilon,2=\rating,3=x,4=z]{\ensuremath{\exp\left(\frac{#1}{2\Delta} #2(#3,#4) \right)}}
\newcommandx{\smint}[4][1=\epsilon,2=\rating,3=x,4=z]{\ensuremath{\frac{\smexp[#1][#2][#3][#4]}{\int_{#4'} \smexp[#1][#2][#3][#4']d #4'}}}
\newcommandx{\smsum}[4][1=\epsilon,2=\rating,3=x,4=z]{\ensuremath{\frac{\smexp[#1][#2][#3][#4]}{\sum_{#4'} \smexp[#1][#2][#3][#4']}}}

\providecommand{\abs}[1]{\ensuremath{{\lvert#1\rvert}}}

\providecommand{\norm}[1]{\ensuremath{\lVert#1\rVert}}

\DeclareMathOperator{\Ima}{Im}

\newcommand{\algname}{SynTF\xspace}

\begin{document}

% LICENSE info from ACM rightsreview
\copyrightyear{2018} 
\acmYear{2018} 
\setcopyright{none} %none
\acmConference[SIGIR '18]{The 41st International ACM SIGIR Conference on Research \& Development in Information Retrieval}{July 8--12, 2018}{Ann Arbor, MI, USA}
\acmBooktitle{SIGIR '18: The 41st International ACM SIGIR Conference on Research \& Development in Information Retrieval, July 8--12, 2018, Ann Arbor, MI, USA}
\acmPrice{15.00}
\acmDOI{10.1145/3209978.3210008}
\acmISBN{978-1-4503-5657-2/18/07}

\settopmatter{printacmref=false, printccs=false, printfolios=true} % We want page numbers on submissions
\fancyhead{}
\renewcommand\footnotetextcopyrightpermission[1]{} % removes footnote with conference information in first column

% TITLE
\title[\algname: Synthetic and Differentially Private Term Frequency Vectors]
{\algname: Synthetic and Differentially Private Term Frequency Vectors for Privacy-Preserving Text Mining}
%\title{SIG Proceedings Paper in LaTeX Format}
%\titlenote{This is an extended version of our paper}
\subtitle{Extended Version}
\subtitlenote{This report is an extended version of our paper of the same name,
	to appear in the proceedings of
	\textit{SIGIR ’18: The 41st International ACM SIGIR Conference on Research \& Development in Information Retrieval}.}

%AUTHORS
\author{Benjamin Weggenmann}
%\authornote{The secretary disavows any knowledge of this author's actions.}
\affiliation{%
	\institution{SAP Security Research}
	%	\streetaddress{P.O. Box 1212}
%	\city{Karlsruhe, Germany} 
	%	\state{Ohio} 
	%	\postcode{43017-6221}
}
\email{benjamin.weggenmann@sap.com}

\author{Florian Kerschbaum}
%\authornote{The secretary disavows any knowledge of this author's actions.}
\affiliation{%
	\institution{University of Waterloo}
	%	\streetaddress{P.O. Box 1212}
%	\city{Waterloo, ON, Canada} 
	%	\state{Ohio} 
	%	\postcode{43017-6221}
}
\email{florian.kerschbaum@uwaterloo.ca}

\renewcommand{\shortauthors}{B.~Weggenmann and F.~Kerschbaum}

%\begin{abstract}
%This paper provides a sample of a \LaTeX\ document which conforms,
%somewhat loosely, to the formatting guidelines for
%ACM SIG Proceedings.\footnote{This is an abstract footnote}
%\end{abstract}
\begin{abstract}

Text mining and \acrlong{IR} techniques
have been developed to
assist us with analyzing, organizing and retrieving documents
with the help of computers.
In many cases, it is desirable that the authors of such documents remain anonymous:
Search logs can reveal sensitive details about a user,
critical articles or messages about a company or government
might have severe or fatal consequences for a critic,
and negative feedback in customer surveys might negatively impact business relations
if they are identified.
Simply removing personally identifying information from a document is, however,
insufficient to protect the writer's identity:
Given some reference texts of suspect authors,
so-called \emph{authorship attribution} methods
can reidentfy the author from the text itself.

%Previous countermeasures to obfuscate texts had only limited success against
%advanced authorship attribution techniques.
One of the most prominent models to represent documents
in many common text mining and \acrlong{IR} tasks is the \acrlong{VSM}
where each document is represented as a vector,
typically containing its term frequencies or related quantities.
We therefore propose an automated text anonymization approach
that produces synthetic term frequency vectors for the input documents
that can be used in lieu of the original vectors.
%in many common text mining and \gls{IR} tasks.
%with little impact on the accuracy of the result,
We evaluate our method on an exemplary text classification task
and demonstrate that it only has a low impact on its accuracy.
In contrast, we show that our method strongly affects authorship attribution techniques
to the level that they become infeasible with a much stronger decline in accuracy.
Other than previous authorship obfuscation methods, %\cite{kacmarcik2006obfuscating},
%yields synthetic representations of the original feature vectors
our approach %to derive synthetic feature vectors for textual data
is the first that fulfills \emph{differential privacy} %the notion of 
and hence comes with a provable plausible deniability guarantee.
%is the first \dyp mechanism for anonymizing textual data.

\end{abstract}




\maketitle

\section{Introduction}
%\subsection{Motivation}
\label{sec:introduction}

%With the invention of the computer, 
%many aspects of our daily lives have moved from the analog to the digital world,
%%huge amounts of textual data have moved from the analog to the digital world,
%and with the advent of the Internet, many of them take place \emph{online}.
%Today, in addition to printed media,
%we have an (almost) infinite source of news and blog articles online,
%we can communicate text instantly through emails and instant messages,
%we retrieve information from online lexica and search any kind of information online.
%We find job positions and file applications,
%receive and process surveys, store and analyze medical records.

%\paragraph*{CONTEXT}
%For centuries, text has been used to convey information between human beings:
%We have been sharing and communicating facts and stories through books, letters, newspapers and magazines.
%The invention of computers and the advent of the Internet have brought us additional possibilities:
%We can now read news online, query search engines and encyclopedias, and communicate instantly through email and instant messages.
%The invention of the computer and the advent of the internet has made textual data ubiquitous:
%We are surrounded
%Furthermore, more and more textual data is being (pre-)processed by machines.
For centuries, text has been used to convey information between human beings
through books, letters, newspapers and magazines.
%Textual data has become ubiquitous, and 
%Since the invention of computers %and the advent of the Internet,
With the advent of the digital age, 
more and more textual data is being processed and analyzed by machines.
Typical tasks include %, for instance,
%\begin{itemize}
%	\item
	text classification, which is used in particular for spam filtering \cite{sahami1998bayesian}
	and automated email routing \cite{busemann2000message},
%	\item
	document retrieval \cite{salton1975vector},
	where indexed documents are retrieved and ranked according to search queries,
%	\item
	sentiment analysis \cite{liu2012sentiment},
%\end{itemize}
and a wide variety of other tasks in the \gls{IR} and text mining domains.
\TODO{less prose below}

In many cases, it is desirable for an author that his writings stay anonymous.
This could be the case if the textual data contains sensitive information about the author,
for instance in search queries.
%or if knowing the author of a text has negative consequences for him,
Negative feedback from customer surveys might negatively impact business relations
if the author or his company is known, and critical news or blog articles
about a company (or government) might have severe (or fatal) consequences
for the author of the article.
In other areas, anonymity is required for compliance or legal reasons,
e.g. in the selection of job candidates to eliminate discrimination.
Furthermore, without anonymity people and data owners might feel reluctant to participate
in surveys or to release their data. Offering anonymity might be a means to convince them
to share their data in an anonymized form, which could then be used to perform evaluations
and as training data for machine learning models.

%not only can reveal your identity, but also your healt status and even thoughts.
%critical blog articles about a government (or boss)
%negative feedback from surveys might negatively impact business relations

Traditional sanitization approaches for free text include removing parts
containing \gls{PII} such as the author's name, or replacing it with a pseudonym.
However, these methods are insufficient to protect the author's identity:
As the famous Netflix de-anonymization attack \cite{narayanan2008robust}
and other studies \cite{rao2000can,sweeney2000simple,jawurek2011smart,de2013unique}
have shown,
the \emph{originator} of data can be \emph{re-identified from the data itself}.
%\TODO{mention prominent examples by name, e.g. netflix. state ``also holds for text''.}
%
We illustrate this in the case of the AOL search data release \cite{barbaro2006face},
where search queries of over 650,000 users were released for research purposes in 2006.
%of Thelma Arnold \cite{barbaro2006face}, a AOL user,
%whose history of search queries were released for research purposes
%together with those of over 650,000 other users in 2006.
The search logs were \enquote{anonymized} by linking the queries to %their users through
a numerical identifier instead of the actual user name.
After some investigation in the data,
the New York Times eventually learned enough information about user 4417749
so they could re-identify her as Thelma Arnold, a 62-year-old widow from Lilburn,
a city in Georgia.

The task of attributing authorship of an anonymous or disputed document to its respective author
is called \emph{authorship attribution}.
%is a set of methods concerned with attributing authorship of anonymous or disputed documents to their authors.
Such methods usually make use of stylistic features to identify or discriminate authors,
as has been done with the statistic techniques in \cite{mosteller1963inference} to resolve the dispute of the Federalist Papers.
Recently, more sophisticated methods have evolved that use statistical analysis
and machine learning to tackle the problem. %and the power of computers 
A survey of modern authorship attribution methods is given by Stamatatos \cite{stamatatos2009survey},
and examples include the JGAAP \cite{juola2006prototype} and JStylo \cite{mcdonald2012use} frameworks.
While these powerful methods are useful in the literary world and in forensics,
they can often pose a threat to the privacy and integrity of authors of documents
with potentially sensitive content.

% TODO Address METAREVIEW
\paragraph*{Contributions}
Many \gls{IR} and text mining algorithms rely on the \gls{VSM} \cite{salton1975vector}
where documents are represented as vectors containing, for instance,
their \gls{tf} or \gls{tfidf} values.
Therefore, we propose a solution that targets this representation
and produces synthetic \acrshortpl{tfvec} which can be used as a substitute
for the original ones.
More precisely, we make the following contributions:
\begin{itemize}
	\item In \cref{sec:approach}, we propose ``\algname'', a \dyp method to compute anonymized,
	\emph{synthetic \acrlongpl{tfvec}} for textual data that can be used as feature vectors
	for common \gls{IR} and text mining tasks such as text classification. %, document retrieval, and sentiment analysis,
%	but makes it hard for attackers to identify the author of the texts based on authorship attribution methods.
	%To our knowledge, this is the first such mechanism fulfilling \dlp.
	%for text documents and textual data that generates synthetic feature vectors suitable for these tasks.
	
	\item In \cref{sec:dp-proofs}, we give theoretical results on the \dlp properties of our method.
	We derive improved bounds for the privacy loss of our method
	and give a heuristic argument that
	\dlp on large (discrete) output spaces demands a large privacy loss
	if the result should fulfill a minimum usefulness requirement.
%	Moreover, we present an alternative bound on the privacy loss for the Exponential mechanism which,
%	in certain scenarios, can be shown to be strictly lower than the given privacy parameter $\epsilon$.
%	It is based on the maximum change across \textit{all outputs} for \textit{fixed inputs}
%	(in contrast to the sensitivity which tracks the maximum change across \textit{adjacent inputs}
%	for \textit{fixed outputs}).
	
	\item In \cref{sec:results},
	we experimentally verify our method on a corpus of newsgroups postings:
	A benign, well-intended analyst wants to classify the documents into certain topics,
	%or categories,
	whereas a malicious attacker tries to re-identify the author of these documents
	using authorship attribution techniques.
	The results show that our method
	has a much stronger impact on the attacker's %(authorship attribution)
	than on the analyst's task. %(topic classification).
	
%	\item Lastly, in \cref{sec:results:necessary}, we give a heuristic argument that
%	differential privacy on large (discrete) output spaces demands a large privacy loss
%	if the result should fulfill a minimum usefulness requirement.
%	This applies particularly to cases where the inputs pertaining to individual persons
%	are \emph{not aggregated}, but anonymized \emph{individually}.
\end{itemize}
Based on our motivation and results,
we presume that the synthetic \acrlongpl{tfvec} (SynTF vectors)
can be used in a multitude of text mining and \gls{IR} tasks
where the semantic similarity of documents is decisive.
On the other hand, our method obliterates stylistic features
that could otherwise reveal the identity and other privacy-sensitive information
about the writer such as age or gender.

%\TODO{
%	Problem and Contributions.
%%	1. simple example to introduce problem (and 'how' in a sentence or two)
%	GDPR Motivation?
%%	Masking approaches like scrubadub {\url{http://scrubadub.readthedocs.org/}} often insufficient!
%	Shown in previous works (Thelma, Netflix, etc.), particularly JStylo; also manual changes can be detected.
%	Our method provides better protection against this attack.
%	(2. highlight main idea: done)
%	(3. list contributions: done)
%	4. forward refs introduction
%	}

%Before we describe the details of our contributions and discuss the results,
%we present the necessary preliminaries in the following \namecref{sec:prelim}.
%Lastly. we discuss related work including other anonymization techniques in \cref{sec:related}
%and give our conclusion in \cref{sec:conclusion}.

% TODO Preliminaries: CONCEPTS, TERMS, NOTATION, THEORY that are NECESSARY for this paper on a TECHNICAL level %
\section{Preliminaries} % Preliminaries
\label{sec:prelim}
In this section, %we will introduce and describe the underlying concepts necessary for our approach.
%First, we describe a certain class of attacks on the anonymity of authors of texts.
we briefly describe text classification,
and follow with a more detailed introduction on \dlp.

%Similar re-identification attacks have been performed on other data types:
%In order to develop an improved movie recommendation algorithm, Netflix released the Netflix Prize dataset with movie ratings of roughly 480,000 anonymized Netflix subscribers.
%Using the Internet Movie Database (IMBd) as additional data source,
%the authors of \cite{narayanan2008robust} managed to re-identify the corresponding records of known users.
%Jawurek et. al. \cite{jawurek2011smart} identified households based on the consumption traces collected by smart meters.
%\cite{de2013unique}

%Scenario for text classification where author wants to remain unidentified and revealing DTM is sufficient

%Our idea:
%replace words by synonyms to preserve meaning (used for utility/classification)
%but destroy relevant information (features) for attack/authorship attribution.
%For each word in the input text, we determine a probability distribution on the underlying vocabulary.
%This is done by first assigning a score to each term and setting the probability proportional to the exponentiation of the score.
%Synonyms for an input word will be sampled according to its corresponding distribution from the vocabulary.
%In essence, this is the same as repeatedly executing the Exponential mechanism where the quality function is defined by the scores for each term.

\subsection{Text Classification}
\label{sec:prelim:classification}

Text classification is the problem of assigning a given text to one or more predefined categories.
It has many applications, for instance in the automated sorting and filtering of email messages, 
spam filtering, categorization of news articles, etc.
The problem is typically solved using machine learning techniques.
In the supervised model, a classifier is \emph{trained}
based on a set of documents with known categories
so it can recognize characteristic features in the text
that indicate the right category.
A trained classifier can then \emph{predict} the most likely category
for new texts whose category is unknown.

%For many text mining tasks such as text classification,
%documents are represented in the vector space and \gls{BoW} models:
To obtain a representation corresponding to the vector space and \gls{BoW} models,
documents are transformed into feature vectors where each entry corresponds
to a certain word in an underlying vocabulary.
The process of this transformation is also called \emph{vectorization}.
Two common representations are \emph{term frequency} (tf) vectors
where each entry equals the number of occurrences of the corresponding term in the document,
and the \emph{term frequency -- inverse document frequency} (tf--idf) vectors 
which are derived from the \acrshortpl{tfvec} by also taking the number of documents into account
that contain the corresponding term.
We refer to \cite{schutze2008introduction} for more information on text mining and information retrieval.

\subsection{Differential Privacy}
\label{sec:prelim:diffpriv}

Differential privacy has first been proposed by Dwork et al. \cite{dwork2006calibrating} % in 2006, then still
under the name \emph{$\epsilon$\-/indistinguishability}.
%in the context of releasing answers to queries on databases
%while protecting the privacy of the individuals whose records constitute the database.
It works by releasing noisy answers to the database queries,
where the noisy results on two databases that differ in only a single record
are probabilistically \emph{indistinguishable} up to a multiplicative factor.
%(that depends on $\epsilon$).
%\TODO{Decide mechanism = RVA or ProbMeasure?}
% favor random variables for allowing more convenient and more general notation.
We give some basic terminology and results as required in the paper.
For a broader introduction and further details on \dlp,
we refer the reader to the book by Dwork and Roth \cite{dwork2014algorithmic}.
We follow the notation of \cite{chatzikokolakis2013broadening},
with the one deviation that we describe random mechanisms via random variables
instead of probability measures on the output space.
Since every random variable induces a probability measure on the underlying space, %in turn
the two definitions are equivalent.

\begin{definition}[Randomized mechanism]
	\label{def:randmech}
	Let $\X$ and $\Z$ be two sets where $\Z$ is measurable,
	%and let $\ProbMeasures(\Z)$ be the set of probability measures over $\Z$.
	and let $\RandVars(\Z)$ be the set of random variables on $\Z$.
	A \emph{randomized mechanism} from $\X$ to $\Z$ is a probabilistic function $\M:\X\to\RandVars(\Z)$
	that assigns a random variable on $\Z$ to each input $x\in\X$.
	From an algorithmic point of view,
	we \emph{run} an instance of a randomized mechanism $\M$ on a given input $x$
	by \emph{sampling} a realization $z$ of the random variable $\M(x)$.
	We write this as $z\samplefrom\M(x)$.
\end{definition}
As noted above, each random variable on $\Z$ induces a probability distribution on $\Z$.
A continuous/discrete distribution is typically described by its probability density/mass function (pdf/pmf).
%We denote the probability density/mass function of a random variable $X$ by $f_X$.
By slight abuse of notation, we write $\Pr[X=x]$ for the pdf/pmf of $X$.%$f_X(x)$.

%Two databases $x_1,x_2\in\X$ that differ in exactly one record are called \emph{adjacent};
%these are precisely those whose Hamming distance $\dHamming(x_1,x_2)$ is $1$.
%One important notion in this context is that of \textit{adjacency} between two inputs:
\begin{definition}[Adjacency]\label{def:adjacency}
	Given a metric $\dX$ on the space $\X$,
	we say that two inputs $x_1,x_2\in\X$ are \textit{adjacent (with respect to $\dX$)} if $\dX(x_1,x_2)\leq1$.
	We write this as $x_1\adj_{\dX} x_2$ (or $x_1 \adj x_2$ if the metric is unambiguous).
\end{definition}

\begin{definition}[Differential Privacy]
	\label{def:diffpriv}
	Let $\epsilon>0$ be a privacy parameter.
	A randomized mechanism $\M:\X\to\RandVars(\Z)$ fullfils \textit{$\epsilon$\-/\dlp} %is \emph{$\epsilon$\-/\dyp}
	if for any two adjacent inputs $x_1,x_2\in\X$, and any set of possible outputs $Z \subseteq \Ima(\M)$,
	\[
	\Pr[ \M(x_1) \in Z ] \leq e^\epsilon \cdot \Pr[ \M(x_2) \in Z ]. %\cdot d_h(x,x')
	\]
	%	where $d_h$ is the Hamming distance on $\X$.
	%	In this case, $\epsilon$ is called the \emph{privacy loss} or \emph{leakage} of \M.
	The \textit{privacy loss} of a randomized mechanism $\M$ is the quantity
	\[
		\ell(\M) := \sup_{x_1 \adj x_2} \sup_{Z\in\Ima(\M)} \ln\frac{\pr{\M(x_1)\in Z}}{\pr{\M(x_2)\in Z}}.
	\]
%	where we interpret $0/0=0$. %thus $\epsilon$ is an upper bound for the privacy loss of an $\epsilon$\-/\dyp mechanism.
\end{definition}
Note that $\epsilon$ is an upper bound for the privacy loss,
and hence any randomized mechanism $\M$ with finite privacy loss $\ell(\M)$
also fulfills $\epsilon$\-/\dlp with $\epsilon=\ell(\M)$.
%Note that by definition, $\epsilon$ is an upper bound for the privacy loss of an $\epsilon$\-/\dyp mechanism,
%Therefore, we can prove that an algorithm fulfills \dlp by bounding its privacy loss.

%\paragraph*{Differential Privacy}
Typically, the input space $\X$ models the set of databases
over some domain of values $\Universe$ with $n$ records, i.e. $\X=\Universe^n$.
%Absent values can be modeled by introducing a special null value $\varnothing$.
In the case of textual documents, we adopt the vector space/\gls{BoW} model
where each document $x$ is represented as feature vector over some vocabulary $\Vocab$ of size $\VocSize$.
Since we anonymize each document independently, we assume $\X=\Z=\RealsGEQ0^\VocSize$.
We consider any two texts as adjacent which is the most strict and conservative way to define adjacency.
%For adjacency, we take a very conservative stance and consider any two texts as adjacent,
%While we could also relax adjacency to only include ``similar'' texts,
%it is difficult to justify an appropriate similarity level for adjacent texts.

%(which equals the number of terms in the underlying vocabulary).
%
%A record corresponds to the feature vector of one document, and
%Since our approach is an input perturbation mechanism
%that anonymizes each document independently,
%we can assume that each input $x\in\X$ represents a single document only,
%so $\X=\Universe=\RealsGEQ0^\VocSize$.
%
%This also implies that we follow the \emph{bounded} interpretation of \dlp,
%since we cannot anonymize an empty input $\varnothing$.

%\TODO{Interactive vs. non-interactive, local model?}

%\TODO{660C-2: highlight difference individual vs aggregate statistics.
%	\cite{sala2011sharing,machanavajjhala2011personalized,he2014blowfish}
%}

\TODO{verify per word or per doc privacy.
	explain advantage of per doc (useful in convex combination lemma, conservative),
	then still state per (output) word epsilon?}

\paragraph{The Exponential Mechanism}
A very important and versatile building block for differential privacy is the Exponential mechanism
by McSherry and Talwar \cite{mcsherry2007mechanism}.
It applies to both numerical and categorical data and fulfills $\epsilon$\-/\dlp
as shown in \cite[theorem 6]{mcsherry2007mechanism}.
%For a given input, it randomly yields an output value based on a \enquote{rating} or \enquote{quality}
%that specifies the suitability of each possible output for the given input:
%While we could define sensitivity for the Laplace mechanism as distance on $\Reals^k$,
It requires a \enquote{measure of suitability} for each possible pair of inputs and outputs:
\begin{definition}[Rating function and sensitivity]
	\label{def:qualityfunc}
	A function $\rating:\X\times\Z\to\Reals$ is called a \emph{rating function} from $\X$ to $\Z$.
	The value $\rating(x,z)$ is the \emph{rating} for input $x$ and output $z$.
%	The \emph{similarity score} for $v_i\in\Vin$ and $w_j\in\Vout$ is given by the value of $\score(v_i)_j$.
	The \emph{sensitivity} $\Delta_\rating$ of the rating function $\rating$ is its largest possible difference given two adjacent inputs, over all possible output values:
	$$\Delta_\rating := \max_{z\in\Z}\ \max_{x_1 \adj x_2} \big( \rating(x_1,z)-\rating(x_2,z) \big)$$
%\max_{x'\in U_1(x)}\end{definition}
\end{definition}
	
In our scenario with textual data, the rating function $\rating$ will be bounded to $[0,1]$,
which implies that its sensitivity is $\Delta_\rating \leq1$.
%so the resulting Exponential mechanism will fulfill $(2\epsilon)$\-/\dlp.

%For a given input, the Exponential mechanism defines a probability distribution
%on the output space where the probability of an output for a given input
%is proportional to their exponentiated rating,
%multiplied by a privacy parameter $\epsilon$:
\begin{definition}[Exponential mechanism]
	\label{def:expmech}
	Let $\epsilon>0$ be a privacy parameter, and let $\rating:\X\times\Z\to\Reals$ be a rating function.
	For each $x\in\X$, we define a random variable $\Em(x)$ that is described by the probability density function (pdf)
	\[ 
		\Pr\big[ \Em(x)=z \big] = \smint.
	\]
	Note that a discrete version of the Exponential mechanism for countable $\Z$
	is obtained by replacing the integral with a sum.
\end{definition}
%A \emph{run} of the Exponential mechanism $\Em$ simply picks a sample $z\samplefrom\Em(x)$
%according to the distribution defined by the above pdf.

%Before we state the result that the Exponential mechanism is \dyp, %we need another definition:
%we need to define how to compute the sensitivity.

%\begin{theorem}[\cite{mcsherry2007mechanism}]
%	\label{thm:expmech-diffpriv}
%	The Exponential mechanism $\Em$
%	% with privacy parameter $\epsilon$ and quality function $\rating$
%	fulfills $\epsilon$\-/\dlp.
%	\qed % see reference for proof
%\end{theorem}
%For a proof of the theorem, see e.g. \cite{mcsherry2007mechanism}.

%\paragraph*{Properties of Differential Privacy}
%
%\begin{theorem}[Composition theorem]
%	\label{thm:dp-composition}
%	A sequence of $\epsilon_i$\-/\dyp mechanisms $\M_i$, $i=1,\ldots,n$,
%	applied to the same input gives $(\sum_{i=1}^n \epsilon_i)$-differential privacy.
%\end{theorem}
%
%\TODO{Composition. Postprocessing.}

%\TODO{Basic example?}

%\section{Differentially Private Document-Term-Matrices}
\section{Synthetic Term Frequency Vectors}
\label{sec:approach}

In this section, we first describe the intended usage scenario.
We then take a closer look under the hood of authorship attribution techniques
and derive the basic motivation behind our \algname method.
Finally, we describe our method in detail and present its \dlp properties.

\subsection{Usage Scenario} % System Model
\label{sec:approach:scenario}

Consider a data processor that wishes to share sensitive training data
for machine learning with a third-party analyst.
Feature vectors are sufficient for most machine learning tasks
since they are produced by the analyst in a preprocessing step anyway.
Our method automatically creates anonymized feature vectors
that can be shared with the analyst and which he can use in lieu of his own vectors.

In our present scenario, we are given a set of text documents such as email messages,
job applications or survey results.
%that contain potentially sensitive information about their authors.
The documents shall be analyzed by a (benign) third-party analyst,
who wants to perform a typical text mining task such as text classification.
Our aim is to prevent authorship attribution attacks as described above. %in the introduction.
Therefore, to protect the identity of the authors and prevent re-identification,
we only provide the analyst with synthetic \gls{BoW} feature vectors instead of the original documents.
Email providers and search engines could share anonymized feature vectors
of emails or (aggregted) search queries with advertising networks
to provide personalized ads while protecting their users. %' identities
%personalized ads
%based on the anonymized feature vectors instead of the original email.
%An online marketing provider could, for instance, use anonymized feature vectors
%of email messages or (aggregted) search queries to provide personalized ads
%while protecting their sources' identities.

%As has been shown by several re-identification attacks \cite{sweeney2000simple,barbaro2006face,narayanan2008robust,de2013unique},
%connections between data and its originator (or author) can be made from the data itself,
%even if any \gls{PII} is removed from the documents.
%Therefore, proper anonymization of the documents requires making modifications to their content.
%\TODO{660B-2: automated processing, no manual changes}
%A (benign) analyst wants to classify the anonymous documents into certain topics,
%whereas a malicious attacker tries to re-identify the author of these documents using authorship attribution techniques.
%Our experiments show that the impact of our anonymization technique (measured by precision, recall, and f-score)
%is much more severe on the attack (authorship attribution) than on the analyst's desired utility (classification).
%\cref{sec:results:evaluation}

%\subsubsection{Desired Utility}
%\label{sec:approach:scenarion:utility}

\paragraph*{Attacker Model}
%\label{sec:approach:scenario:attack}
 %\cref{sec:prelim:authorship}.
The attacker is presented with a document of unknown authorship
which has been written by one of several suspected authors.
Her goal is to identify the document's actual author from the group of suspects.
We assume that she has a set of similar reference documents from each suspect
that she can use to help decide which suspect to assign the unknown document to.

We compare the attacker's capability to re-identify the authors
on the original plaintexts as well as the anonymized feature vectors. %documents.
%The latter are given only as synthetic feature vectors.
We assume the attacker knows the dictionary,
so she can convert the numbers in the feature vectors
to a textual representation by repeating each word accordingly.
This allows her to (partially) deduce more complex features beyond \gls{BoW},
such as the \acrshort{WP} feature set which is often used in authorship attribution \cite{abbasi2008writeprints,mcdonald2012use}.
As explained in the next \cref{sec:prevent-aa},
most of these features cannot be correctly inferred anymore,
%as they are lost in the synthetization process,
which is beneficial for our method as these are precisely the stylistic features
(beyond \gls{BoW}) that are exclusively exploited by our attacker.

%\TODO{660A-2: Describe knowledge of vocabulary for WritePrints derivation.}

\subsection{Preventing Authorship Attribution}
\label{sec:prevent-aa}
%Typical authorship attribution methods make use of different kinds of features.
A popular feature set for authorship attribution has been described
in the WritePrints method \cite{abbasi2008writeprints}.
It includes the following types of stylistic features:
\begin{description}
	\item[Lexical] Counts of letters, digits, special characters, number of characters and words, etc.
	\item[Syntactic] Frequency of function words, punctuation, parts of speech (POS) tags.
	\item[Structural] Number and length of paragraphs and sentences, URLs or quoted content, etc.
	\item[Content] Frequencies of words (\gls{BoW} model).
	\item[Idiosyncratic] Misspelled words.
\end{description}
For some features such as letters, words, digits and POS tags, it also considers their bi- and trigrams,
thus taking \emph{order information} into account.
These features have a strong capability to capture individual stylistic characteristics
expressed by the writer of a text. For instance, one author might subconsciously
prefer using the passive voice or past tense, so many verbs will end in an ``ed''-bigram,
whereas another author might tend to use the present continuous or gerund
which causes many ``ing''-trigrams.

Ordinary text mining and \gls{IR} tasks such as classification typically only use content-level features
which are often modeled and represented as \glspl{tfvec} in the \gls{BoW} model.
Most of the stylistic features used for authorship attribution thus get \emph{lost in vectorization}:
In fact, the tf \emph{vectors} by their very nature
do not capture any structural information, and most syntactic features will be destroyed as well.
Apart from the content (and idiosyncratic) features, however, we can still derive lexical
features if the \gls{BoW} vocabulary is known.

Since the attacker can still exploit the derived lexical features,
we aim at disturbing them in a way that keeps the meaning or theme of a document intact,
thus further allowing the classification task but impairing authorship attribution.
Lexical features are mostly related with the spelling, therefore,
our idea is to replace words in the input with words with similar meaning (synonyms) but different spelling
%The hope is that these substitute words differ enough in spelling
to make the lexical features meaningless for the attacker.
On the other hand, this will preserve the general theme of the text,
so we hope that the impact is little on the classification task.

\subsection{The \algname Mechanism}
\label{sec:approach:dp-shuffle}

Our goal is a \dyp anonymization method to derive synthetic feature vectors
that keeps the theme of the represented document intact
and at the same time prevents authorship attribution attacks.
%As an additional technical requirement,
For performance and memory efficiency reasons,
we require our method to preserve the sparseness in the \glspl{tfvec}.
Simply applying Laplace noise \cite{dwork2006calibrating} or \dyp histogram publication methods \cite{xu2013differentially}
will fail this requirement, since they produce dense vectors.
%
%To achieve these goals, we propose an approach to produce synthetic \glspl{tfvec}.
Our core idea is to take a word count entry for one term in the \gls{tfvec}
and probabilistically distribute it across all terms in the pre-defined vocabulary.
The probability of each term is determined according to its \emph{similarity} with the original word.
Word similarity can be expressed in various ways, cf. \cref{sec:params:similarity-functions}. %which we will discuss later in
\TODO{More concise discussion of dp implications.}

Differential privacy presents a strong requirement for the method:
Namely, \emph{every} possible output must occur with non-zero possibility for \emph{any} other input.
%It must be possible that two completely unrelated texts produce the same output. 
%The method must produce any valid output from every possible input text with non-zero probability.
This means that a statement on food preference can be processed
to the same output as a conversation on politics,
with non-zero probability. %\footnote{assuming the output space is finite}.
%In order to fulfill differential privacy,
%Therefore, \emph{every} possible output must occur with non-zero possibility for \emph{any} other input.
This has two implications:
First, we must ensure that the probability of picking a term
is always greater than zero, even for totally unrelated words.
Second, it must be possible that two input texts of different lengths
produce the same number of words in their resulting \glspl{tfvec}.
Therefore, we must also specify the output length.
Note that this approach limits the number of entries
that are changed from the original to the anonymized \gls{tfvec},
so it keeps the sparseness of the resulting vector intact.

%Approach:
%\begin{enumerate}
%	\item Fix vocabulary, vectorize document
%	\item For each word in document:
%	\begin{itemize}
%		\item Count frequency of each word in document (including its context)
%		\item Derive probability distribution on vocabulary proportional to similarity score to word (and suitability of term for context)
%	\end{itemize}
%	\item Choose new \enquote{length} for anonymized document (i.e. row sum in DTM)
%	\begin{itemize}
%		\item Fixed length or choose randomly (close to original length) from a pre-defined interval.
%	\end{itemize}
%	\item Repeat for each word according to its frequency in original document:
%	\begin{itemize}
%		\item Pick synonym to word from vocabulary according to its distribution
%	\end{itemize}
%\end{enumerate}

\paragraph{Algorithm Description}
\label{sec:approach:dp-shuffle:algo}

In the following, let $V$ denote the underlying vocabulary of size $\abs{\Vocab}=\VocSize$.
The vocabulary could be derived, for instance, from a reference corpus of documents
from a similar context as the target documents which shall be anonymized.
We will describe the \algname approach for a single document $T$,
but it is possible to anonymize an entire corpus simultaneously.
The anonymization for a document $T$ consists of two main phases:
\begin{description}
	\item[Analysis]
	We \emph{vectorize} the document $T$ to its feature vector $\docvec=(\doc_1,\ldots,\doc_K)\in\RealsGEQ0^K$.
%	with non-negative entries.
	Typically, $\docvec$ will be a tf or \acrshort{tfidfvec} over the underlying vocabulary $\Vocab$.
	Next, we \emph{normalize} $\docvec$ with respect to the $\ell_1$-norm
	to transform it into a \emph{composition vector} $\compvec_\docvec:=\docvec/\norm{\docvec}_1$
	whose entries can be interpreted as probability distribution over $\Vocab$.
%	that is, $0\leq \theta_{t,i} \leq 1$ and $\sum_i \theta_{t,i} = 1$.
	
	\item[Synthesis] 
	We repeatedly sample terms $v_1,\ldots,v_n$ from the distribution $\compvec_\docvec$ on $\Vocab$.
	For each $v_i$, we use the Exponential mechanism to pick a substitute output term $w_i\in\Vocab$
	with probability proportional to a \emph{similarity rating} $\rating(v_i,w_i)$.
	Finally, we construct a synthetic \gls{tfvec} $\mathbf{s}\in\Nats_{\geq0}^\VocSize$
	of length $n$ by counting all the terms $w_i$.% for $i=1,\ldots,n$.
\end{description}

\Cref{alg:synth-tf} illustrates the synthesis phase of our \algname mechanism in pseudocode.
%For completeness, we state the following definition which is used in the code:
It uses the following definition:
\begin{definition}[Categorical distribution]
	\label{def:categorical}
	For an enumerable set $V=\set{v_1,\ldots,v_k}$ %with $k$ elements
	and associated probability vector $\mathbf{p}=(p_v)_{v\in V}$ with $\sum_{v\in V} p_v=1$,
%	For $k\in\NatsG0$ and event probabilities $0\leq p_1,\ldots,p_k\leq1$ with $\sum_i p_i=1$,
	the \emph{categorical distribution}, denoted $\Cat(\mathbf{p})$,
	is defined on $V$ through $\pr{\Cat(\mathbf{p})=v_i}=p_i$.
	%that is, each $v\in V$ occurs with probability $p_v$.
%	For $k\in\NatsG0$ and event probabilities $0\leq p_1,\ldots,p_k\leq1$ with $\sum_i p_i=1$,
%	the \emph{categorical distribution}, denoted $\Cat(p_1,\ldots,p_k)$,
%	is defined on $\{1,\ldots,k\}$ by the probability mass function $f(i)=p_i$,
%	that is, each $i$ occurs with probability $p_i$.
\end{definition}

\begin{algorithm}
	%	\KwData{scoring function $\sigma$} %:\Vin\to\interval1^m
	\KwIn{document vector $\compvec_\docvec$, desired output length $\reslen$,
		privacy parameter $\epsilon>0$, rating function $\rating:\Vocab\times\Vocab\to[0,1]$}
	\KwResult{synthetic \gls{tfvec} $\resvec\in\Nats^{\abs{\Vocab}}$ with $\abs{\resvec}=\reslen$} %of length $\abs{\resvec}=\reslen$
	\BlankLine
	\For(\tcp*[f]{produce output term-by-term}){$i \leftarrow 1$ \KwTo $n$}{
		$v_i \samplefrom \Cat(\compvec_\docvec)$\tcp*{sample word $v_i$}
		$w_i \samplefrom \Em(v_i)$\tcp*{choose synonym for $v_i$}
	}
	%	$\res \leftarrow \sum_{i=1}^n w_i$\tcp*{sum up synonym counts}
	$\resvec \leftarrow \parens[\big]{\abs{\set{i\in\irange{n}:w_i=w}}}_{w\in\Vocab}$\tcp*{count synonyms}
	\caption{\algname Term-Frequency Vector Synthesis}
	\label{alg:synth-tf}
\end{algorithm}

\subsection{Differential Privacy Results}
\label{sec:dp-proofs}
\todo{Omit proofs.}

In this section, we will prove that our \algname mechanism fulfills \dlp (\cref{def:diffpriv}),
which amounts to deriving an upper bound $\epsilon$ on its privacy loss.

We keep the previous notation where $\Vocab$ is the vocabulary of size $\VocSize$,
%from the previous section 
$\docvec=(\doc_1,\ldots,\doc_K)$ is the tf or \acrshort{tfidfvec} of the target document to be anonymized,
and $\compvec_\docvec:=\docvec/\norm{\docvec}_1$ is the corresponding vector of probabilities.
For each pair of words $v,w\in\Vocab$, we have a similarity score $\rating(v,w) \in [0,1]$.
This score will be used in the Exponential mechanism,
which outputs $w$ on input $v$ with probability
\[
	\pi_{v,w} := \pr{\Em(v) = w} = \smsum[\epsilon][\rating][v][w].
\]
Note that we assume that \textit{all} potential inputs are adjacent
which is a very conservative interpretation of \dlp. % TODO local
This is used in the following lemma, which presents a (niched) counterpart
to the known postprocessing lemma \cite[proposition 2.1]{dwork2014algorithmic},
to show that a convex combination of an $\epsilon$\-/\dyp algorithm is again $\epsilon$\-/\dyp.
\begin{lemma}[Randomized Preprocessing]
	\label{lemma:preprocessing}
	Given two independent randomized mechanisms
	$\A:\X\to\RandVars(\Y)$ and $\B:\Y\to\RandVars(\Z)$,
	%	where $\A$ is \dyp for $\epsilon>0$,
	we define their functional composition $\B\circ\A: \X\to\RandVars(\Z)$
	as first sampling from $\A$ and using the resulting sample as input for $\B$.
	The composition $\B\circ\A$ is $\epsilon$\-/\dyp provided that $\B$ is $\epsilon$\-/\dyp
	where \emph{all} inputs $y,y'\in\Y$ to $\B$ are considered adjacent
	(that is, $\dist_\Y(y,y')\leq1$). %$\forall y,y'$
%	By composition, we understand first sampling from $\A$
%	and using the resulting sample as input for $\B$.
\end{lemma}
\begin{proof}
	Define $\alpha_{x,y} := \pr{\A(x)=y}$ and $\beta_{y,z} := \pr{\B(y)=z}$ for all $x\in\X,y\in\Y,z\in\Z$.
	Then
	\begin{align*}
		\Pr\brackets*{\parens{\B \circ \A}(x) = z}
		&= \Pr\parens*{\bigsqcup_{y\in\Y} \brackets{\A(x)=y} \land \brackets{\B(y)=z}} \\
		&= \sum_{y\in\Y} \pr{\A(x)=y} \cdot \pr{\B(y)=z} \\
		&= \sum_{y\in\Y} \alpha_{x,y} \cdot \beta_{y,z}.
	\end{align*}
	The first equality stems from enumerating, over $y\in\Y$,
	all possible ways to get output $z$ on input $x$.
	The second equality is due to the fact that these possibilities are disjoint,
	and uses the independence between the two randomized mechanisms.
	
	Fix any adjacent $x_1,x_2\in\X$ and $z\in\Z$ and define the quantities
	$\hat\beta_z:=\max_{y\in\Y}\beta_{y,z}$ and $\check\beta_z :=\min_{y\in\Y}\beta_{y,z}$.
	Now
	\begin{align*}
		\frac{ \Pr\brackets*{\parens{\B \circ \A}(x_1) = z} }{ \Pr\brackets*{\parens{\B \circ \A}(x_2) = z} }
		= \frac{ \sum_{y} \alpha_{x_1,y} \beta_{y,z} }{ \sum_{y} \alpha_{x_2,y} \beta_{y,z} }
		\leq \frac{ \sum_{y} \alpha_{x_1,y} \hat\beta_z }{ \sum_{y} \alpha_{x_2,y} \check\beta_z }
		= \frac{\hat\beta_z}{\check\beta_z}
		\leq e^\epsilon,
	\end{align*}
	since the sums are convex combinations of $\beta_{y,z}$
	and since both values of $y\in\Y$ that maximize/minimize $\beta_{y,z}$ are adjacent.
\end{proof}

We now show that \cref{alg:synth-tf} is \dyp:
\begin{theorem}[Differential Privacy of \algname]
	\label{thm:shuffle-dp}
	Given a privacy parameter $\epsilon>0$ and an output length $\reslen\in\Nats$,
	our \algname mechanism (\cref{alg:synth-tf}) fulfills $\epsilon \reslen$-differential privacy.
	\qed
\end{theorem}
\begin{proof}
	Each iteration (the body of the \texttt{for}-loop) consists of two steps: % in \cref{alg:synth-tf}),
	First, our algorithm samples one word $v$ according to the probabilities in $\compvec_\docvec$,
	which can be thought of running a randomized mechanism $\A$
	with the underlying categorical distribution.
	Second, it substitutes $v$ with another word $w\in\Vocab$
	according to their similarity using the Exponential mechanism $\Em$,
	which provides $\epsilon$\-/\dlp.
	By the preceding \cref{lemma:preprocessing}, both steps combined are $\epsilon$\-/\dyp.
	Since we iterate $\reslen$ times, the sequential composition theorem
	\cite[theorem 3.16]{dwork2014algorithmic}
	yields $\epsilon \reslen$\-/\dlp for the entire \texttt{for}-loop.
	Aggregating the synonym counts is a simple postprocessing step
	which keeps the privacy loss unchanged \cite[proposition 2.1]{dwork2014algorithmic},
	and we hence achieve $\epsilon \reslen$\-/\dlp for \cref{alg:synth-tf}.
	%This is precisely a functional composition of a randomized mechanism that picks 
\end{proof}
%The proof uses a counterpart of the known postprocessing lemma
%\cite[proposition 2.1]{dwork2014algorithmic},
%which states that a convex combination of an $\epsilon$\-/\dyp algorithm
%is again $\epsilon$\-/\dyp.

\subsubsection{Alternative Bound for the Exponential Mechanism}
We can derive an alternative bound for the privacy loss of the Exponential mechanism
by also considering the maximum change across \textit{all outputs} for \textit{fixed inputs}
(in contrast to the sensitivity which tracks the maximum change
across \textit{adjacent inputs} for \textit{fixed outputs}):

\begin{theorem}[Alternative bound] % for the Discrete Exponential Mechanism
	\label{thm:shuffle-dp:alternate}
	Let $\epsilon>0$ be a privacy parameter
	and $\rating:\X\times\Z\to\Reals$ be a rating function with sensitivity $\Delta$
	where $\abs{\Z}=L$. %\in\Nats$.
	Let $$\bar\Delta := \max_{x\in\X} \max_{z,z'\in\Z} \abs{\rating(x,z)-\rating(x,z')}.$$
	Then the privacy loss $\ell(\Em)$ is bounded by $(\bar\epsilon + \ln\eta)$,
	%proving $(\bar\epsilon + \ln\eta)$\-/\dlp,
%	\begin{equation*}
%		\frac{\exp(\bar{\epsilon}/2)+L-1}{\exp(-\bar{\epsilon}/2)+L-1} = \bar\epsilon + \ln\eta
%	\end{equation*}
	where
	\begin{equation*}
		\bar{\epsilon} := \epsilon\frac{\bar\Delta}{\Delta}
		\quad\text{and}\quad
		\eta = \eta(\bar\epsilon,L)
		= \frac{e^{-\bar{\epsilon}/2}+L-1}{e^{\bar{\epsilon}/2}+L-1}
	%	= \frac{\exp\parens*{-\frac{\epsilon\bar{\Delta}}{2\Delta}}+L-1}
	%	       {\exp\parens*{\frac{\epsilon\bar{\Delta}}{2\Delta}}+L-1}
		<1.\qed
	\end{equation*}
\end{theorem}
\begin{proof}
	For any $x\in\X$ and $z\in\Z$, denote by $$\pi_{x,z} := \pr{\Em(x) = z} = \smsum[\epsilon][\rating][x][z]$$
	the probabilities that $\Em$ outputs $z$ on input $x$.
	Then for adjacent $x_1,x_2\in\X$ and any fixed $z\in\Z$, we bound
	\begin{align*}
		\frac{\pi_{x_1,z}}{\pi_{x_2,z}}
		&= \smsum[\epsilon][\rating][x_1][z]\cdot \parens*{\smsum[\epsilon][\rating][x_2][z]}^{-1}\\
		&= \frac
		{ \sum_{z'} \exp\parens[\big]{\frac{\epsilon}{2\Delta} \brackets*{\rating(x_2,z')-\rating(x_2,z)}} }
		{ \sum_{z'} \exp\parens[\big]{\frac{\epsilon}{2\Delta} \brackets*{\rating(x_1,z')-\rating(x_1,z)}} } \\
		&\leq \frac
		{ 1 + \sum_{z'\neq z} \exp\parens[\big]{\frac{\epsilon}{2\Delta} \brackets{\overbracket{\rating(x_2,z')-\rating(x_2,z)}^{\leq\bar{\Delta}}}} }
		{ 1 + \sum_{z'\neq z} \exp\parens[\big]{\frac{\epsilon}{2\Delta} \brackets{\underbracket{\rating(x_1,z')-\rating(x_1,z)}_{\geq-\bar{\Delta}}}} } \\
		&\leq \frac
		{ 1 + (L-1) \exp\parens[\big]{\frac{\bar{\epsilon}}{2}} }
		{ 1 + (L-1) \exp\parens[\big]{-\frac{\bar{\epsilon}}{2}} }
		= e^{\bar{\epsilon}} \cdot \underbracket{\frac
		{ e^{-\bar{\epsilon}/2} + L-1 }
		{ e^{\bar{\epsilon}/2} + L-1 } }_{=:\eta < 1}.
%		= e^{\bar{\epsilon}} \cdot \underbracket{\frac
%		{ \exp\parens*{-\frac{\bar{\epsilon}}{2}} + L-1 }
%		{ \exp\parens*{\frac{\bar{\epsilon}}{2}} + L-1 } }_{=:\eta < 1}.
	\end{align*}
	The result follows by taking logarithms and observing that
	the numerator for $\eta$ is strictly smaller than its denominator.
\end{proof}

Note that in general cases we normally have $\bar\Delta > \Delta$
since the sensitivity $\Delta$ is restricted to \textit{adjacent} inputs.
The growth due to the factor $\bar{\Delta}/\Delta$ in $\bar{\epsilon}=\epsilon\bar{\Delta}/\Delta$
would therefore typically exceed the savings due to $\ln\eta<0$,
%in the alternative bound $\bar{\epsilon}+\ln\eta$, which
so the alternate bound $\bar{\epsilon}+\ln\eta$ would be \textit{worse}
than the original bound $\epsilon$ as derived in the standard \dlp proof
for the Exponential mechanism \cite{mcsherry2007mechanism}.
However, if we consider all inputs as adjacent, and if $\rating$ is symmetric in its arguments,
then we will have $\bar\Delta=\Delta$ and $\bar{\epsilon}=\epsilon$,
and thus the factor $\eta<1$ will provide a real improvement over the original bound.
This is the case in our algorithm:
\begin{corollary}[Improved \dlp bound]
	\label{cor:shuffle-improved}
	Given a privacy parameter $\epsilon>0$ and an output length $\reslen\in\Nats$,
	our \algname mechanism %(\cref{alg:synth-tf})
	fulfills $\parens*{(\epsilon+\ln\eta(\epsilon,\VocSize))\cdot\reslen}$\-/\dlp.
%	where $\VocSize=\abs{\Vocab}$ is the size of the vocabulary.
	\qed
\end{corollary}
\begin{proof}
	The proof is identical to that of \cref{thm:shuffle-dp}
	with the improvement that the Exponential mechanism provides
	$(\epsilon+\ln\eta(\epsilon,\VocSize))$\-/\dlp.
\end{proof}

We illustrate the effects of the factor $\eta(\epsilon,\VocSize)$ in \cref{fig:theory:improved-bound}:
The original upper bound $\epsilon$ is the black dotted line on top,
the other lines show the improved upper bound $\epsilon+\ln\eta$
for different values of $\VocSize\in\set{2,100,30000}$.
30000 is approximately the size of the vocabulary in some of our experiments,
The effect of the improved bound increases with the privacy parameter $\epsilon$,
whereas large output spaces have a smoothing effect that dampens the improvement.

%The size of the output space $\VocSize$ acts as a smoothin gfactor
%so the effect becomes less profound for larger $\VocList$.
%We observe a larger upper bound if we increase the size of the output space $\VocSize$
%(i.e. the size of the vocabulary in our scenario).
%With increasing $\epsilon$, however, the size of the output space becomes less important,
%whereas the improvement against the original bound increases.

%The reduction of the upper bound for the privacy loss
%grows larger with increasing values of $\epsilon$,

\begin{figure}%[tbh]
	\centering
	\includegraphics[width=0.83\columnwidth]{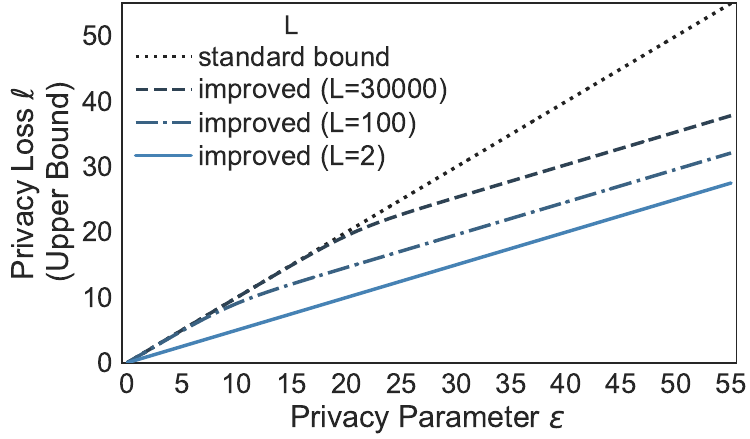}
	\caption{Standard and alternative upper bound $\epsilon+\ln\eta$ for the privacy loss $\ell(\Em)$
		given different output space sizes $L$.}
%		for different values of $\epsilon$ and the size $K$ of the output space.}
	\label{fig:theory:improved-bound}
\end{figure}

\subsubsection{Tight Worst-Case Bounds} %based on Known Parameters}
\label{sec:shuffle:tight-bounds}
A major factor in the \dlp proof of \cref{thm:shuffle-dp:alternate} and \cref{cor:shuffle-improved}
consists of bounding the privacy loss $\ell(\Em)$ for the Exponential mechanism used in \cref{alg:synth-tf}.
This privacy loss is defined as smallest upper bound for the fractions ${\pi_{v_1,w}}/{\pi_{v_2,w}}$,
where $\pi_{v,w}\propto\smexp[\epsilon][\rating][v][w]$ are the the associated probabilities.
%From the proof of \cref{thm:shuffle-dp},
%we know that we need to bound $\frac{p_{\doc,w}}{p_{\doc',w}}$
%in order to prove \dlp for a single run of $\M$, %bound the privacy loss and 
%where $p_{\doc,w} = \sum_{v\in\Vocab} \comp_{\doc,v} \cdot \pi_{v,w}$.
The probabilities $\pi_{v,w}$ depend on the underlying vocabulary $\Vocab$,
%=\smexp[\epsilon][\rating][v][w]
the rating function $\rating$, and the privacy parameter $\epsilon$,
but do not take the documents $\docvec$ and $\docvec'$ into account.
Therefore, we can compute the privacy loss
\begin{equation}\label{eq:tight-loss}
	\ell(\Em)= \max_{w\in\Vocab} \frac{\max_{v\in\Vocab}\pi_{v,w}}{\min_{v\in\Vocab}\pi_{v,w}}
\end{equation}
in advance and independently
from any documents to be anon\-y\-mized once the parameters $\Vocab$, $\rating$,
and $\epsilon$ have been determined.
Our \algname method with privacy parameter $\epsilon$ and output length $n$
thus in fact fulfills $\ell n$- instead of $\epsilon n$\-/\dlp where $\ell=\ell(\Em)$
is the privacy loss of the Exponential mechanism.
This turns out to lead to huge gains in practice, reducing the privacy loss upper bound
by almost 50\% in our experiments (cf. \cref{sec:results:evaluation}).
%\TODO{Depict gains from tight bound (experimental) vs. improved bound (theoretical) vs. original bound (theoretical).}

To see that these bounds are tight,
note that we can craft two input documents $t_1$ and $t_2$
that each consist of only a single word $v_1$ and $v_2$, respectively,
where $v_1$ and $v_2$ are precisely those that maximize the fraction
$\frac{\pi_{v_1,w}}{\pi_{v_2,w}}$ in \cref{eq:tight-loss}
for the optimal $w\in\Vocab$.

%Theoretical Results
\subsubsection{Relationship between Utility and Privacy Loss}
%\subsubsection{Necessary Condition on $\epsilon$}
\label{sec:results:necessary}

We present the following theoretical results for the Exponential mechanism
which suggest that in order to get ``useful'' outputs with a large output space,
we need to choose a large privacy parameter $\epsilon$ in the order of $\ln\abs\Z$,
under the assumption that there are only few good outputs for each input.
%in order to get a \emph{useful} result, we need to choose a large privacy parameter $\epsilon$ if the output space is large.

%\TODO{Define usefulness?}
%\begin{definition}[usefulness]
%\end{definition}

\begin{theorem}[Upper and Lower Bounds for Utility]
	\label{thm:utility-bounds}
	Let $\rating:\X\times\Z\to\Reals$ be a rating function with sensitivity $\Delta$,
	and $\abs\Z \in \Nats$,
	Take any fixed $x\in\X$ and denote by $\hat{\rating}_x$ and $\check{\rating}_x$ %$\max_{z\in\Z} \rating(x,z)$
	the maximum and minimum rating scores of any output for $x$, respectively.
	For a desired minimum rating $\tau\in[\check{\rating}_x,\hat{\rating}_x]$,
	%	divide $\Z$ into $\Z_{\geq}:=\set{z\in\Z:\rating(x,z)\geq \tau}$
	split $\Z$ into $\T:=\set{z\in\Z:\rating(x,z)\geq \tau}$ and $\overline{\T}:=\Z\setminus\T$.
	%	and $\Z_{<}:=\Z\setminus\Z_{\geq}$. %with respective sizes $L_\geq$ and $L_{<}$.
	Then the probability $\Pr\brackets*{\Em(x)\in\T}$ that the Exponential mechanism % $\Em(x)$
	yields an element %from $\Z_{\geq}$
	with score at least $\tau$ has lower and upper bounds 
	%	\begin{equation*}
	%		1+\frac{\abs{\overline{\T}}}{\abs{\T} \exp\parens*{\frac{\epsilon c}{2\Delta}}}
	%		\leq 
	%		\Pr\brackets*{\Em(x)\in\T}\inv
	%		\leq
	%		1+\frac{\abs{\overline{\T}}}{\abs{\T} \exp\parens[\big]{\frac{\epsilon\bar{\Delta}}{2\Delta}}},
	%	\end{equation*}
	\begin{equation*}
	\frac{\abs{\T}}{\abs{\T}+\abs{\overline{\T}}\exp\parens[\big]{-\frac{\epsilon c}{2\Delta}}}
	\leq
	\Pr\brackets*{\Em(x)\in\T}
	\leq
	\frac{\abs{\T}}{\abs{\T}+\abs{\overline{\T}}\exp\parens[\big]{-\frac{\epsilon \bar{\Delta}}{2\Delta}}},
	\end{equation*}
	where $c:=\tau - \max_{z\in\overline{\T}} \rating(x,z)$
	is the difference between $\tau$ and the next lower rating score,
	and $\bar{\Delta}:=\hat{\rating}_x - \check{\rating}_x$.
\end{theorem}
\begin{proof}
	For the lower bound, consider the inverse probability
	\begin{align*}
	\Pr\brackets*{\Em(x)\in\T}\inv 
	&= \frac{\sum_{z\in\Z} \smexp}{\sum_{z\in\T} \smexp} \\
	&= 1 + \frac{\sum_{z\in\overline{\T}} \smexp}{\sum_{z\in\T} \smexp} \\
	&\leq 1 + \frac
	{\abs{\overline{\T}}\exp\parens*{\frac{\epsilon}{2\Delta}(\tau-c)}}
	{\abs{\T}\exp\parens*{\frac{\epsilon}{2\Delta}\tau}} \\
	&\leq 1 + \frac{\abs{\overline{\T}}}{\abs{\T}} \exp\parens*{-\frac{\epsilon c}{2\Delta}}.
	\end{align*}
	The upper bound is derived similarly.
\end{proof}

Solving for $\epsilon$, these bounds lead to the following corollary:
\begin{corollary}[Necessary and Sufficient Conditions on $\epsilon$]
	Given a probability $p\in[0,1]$, and with %$c$ and $\bar{\Delta}$ as above,
	the notation from \cref{thm:utility-bounds},
	we have the following necessary and sufficient conditions on $\epsilon$
	for $\pr{\Em(x)\in\T} \geq p$:
	\begin{equation*}
		\epsilon \geq
		\begin{cases}
			\frac{2\Delta}{\bar{\Delta}} \ln\parens[\Big]{\frac{p}{1-p} \cdot \frac{\abs{\overline{\T}}}{\abs{\T}}}
				& \text{(necessary condition)} \\
			\frac{2\Delta}{c} \ln\parens[\Big]{\frac{p}{1-p} \cdot \frac{\abs{\overline{\T}}}{\abs{\T}}}
				& \text{(sufficient condition)}
		\end{cases}
	\end{equation*}
	\qed
\end{corollary}

Note that for our \algname algorithm, we have $\bar{\Delta}=\hat{\rating}_x - \check{\rating}_x\leq \Delta$.
Hence for $p=1/2$, the necessary condition becomes
\begin{equation*}
\epsilon \geq 2\ln\parens[\bigg]{\frac{p}{1-p}\cdot\cfrac{\abs{\overline{\T}}}{\abs{\T}}}
	= 2\ln\parens[\bigg]{\cfrac{\abs{\overline{\T}}}{\abs{\T}}}
	= 2\ln\parens[\bigg]{\cfrac{\abs{\Z}-\abs{\T}}{\abs{\T}}}.
\end{equation*}
%There is at least one element in $\T$ that achieves the maximum score $\hat{\rating}_x$, but
Given a reasonable choice of $\tau$,
the number $\abs{\T}$ of ``useful'' outputs whose score is at least $\tau$ will be small.
In the case of our \algname mechanism, we can think of $\tau$ as a threshold for the rating function
that distinguishes good alternatives for a given word from poor ones,
and $\abs{\T}$ would reflect the number of suitable substitutes (synonyms).
If we assume $\abs{\T}$ to be bounded by some constant, then $\epsilon\in\Omega(\ln \abs{\Z})$,
that is, $\epsilon$ \textit{needs} to grow logarithmically in the size of the output space $\abs{\Z}$
in order to allow meaningful results.

\section{Evaluation}
\label{sec:results}

In this section, we first describe our implementation of the \algname mechanism
along with associated parameters and our implementation choices.
We then describe our experiment setup and report the evaluation results.
Finally, we compare \algname with a traditional information removal approach
in the same experiment setup.
\TODO{660B-6: Include na\"ive baseline with independent Laplace noise?}

\subsection{Algorithm Implementation and Parameters}
\label{sec:implementation}

We implemented a prototype of our \algname algorithm in Python
using the \texttt{SpaCy} package ({\url{http://spacy.io/}}) for text parsing functionality
as well as the \texttt{numpy} and \texttt{SciPy} packages \cite{scipy,walt2011numpy}
for (vector) computations.
%using the ``GloVe'' word vectors \cite{pennington2014glove} to derive word similarity.
%We represented feature vectors as sparse \texttt{SciPy} arrays \cite{scipy,walt2011numpy}.
%TODO Describe tools used in evaluation in the ``evaluation'' section below...
%To evaluate the classification task (utility), we used the \texttt{scikit-learn} machine learning library \cite{scikit-learn}.
%For the attack, we used the \texttt{JStylo} authorship attribution framework \cite{mcdonald2012use}
%from Drexel University's Privacy, Security, and Automation Laboratory\footnote{\url{http://psal.cs.drexel.edu/}}.
%
%Used software:
%\begin{enumerate}
%	\item numpy
%	\item sklearn\footnote{\url{http://scikit-learn.org/}} (utility evaluation) \cite{scikit-learn}
%	\item JStylo with CLI (attack evaluation)
%\end{enumerate}
%
Besides the explicit parameters mentioned in \cref{alg:synth-tf},
there are various implementation-dependent parameters
that influence \algname in its different stages.
We now describe these parameters and corresponding implementation choices.
\TODO{shorten! focus: vectorizer (morphology, synsets); synthesis (length n, rating function: vector similarity + bigram factor s, epsilon)}

\subsubsection{Vocabulary and Vectorization}
\label{sec:params:vocabulary}
%\paragraph{Vectorizer Implementation}
We build a custom vectorizer to extract the vocabulary from the training or a given reference corpus,
and to subsequently transform documents to their \acrshort{BoW} \glspl{tfvec}.
We can specify several special options:
%we can specify additional options as described above: %in \cref{sec:params:vocabulary}:
Firstly, we can choose, for each extracted word,
to keep its spelling as-is, to change its morphology through lemmatization,
or to convert it to lower case.
Secondly, we can instruct the vectorizer to include additional terms
that are similar or synonymous to the actually extracted words,
as to provide a greater choice of candidates for replacing a word with a suitable synonym
but hopefully with different spelling to disturb lexical authorship attribution features.
Our implementation uses the synonyms provided by WordNet's synsets.
We remove stop words and numbers by default.

%In order to represent documents in the \gls{BoW} vector space model,
%we need to define an underlying vocabulary of terms to take into account.
%Typically, the vocabulary is extracted from a reference corpus of documents
%or the training data,
%where several options that influence the extraction process can be specified.
%Some common options include capitalization (e.g. convert all terms to  lower case),
%stemming or lemmatization, the exclusion of certain (stop) words,
%and the treatment of numerical and special characters.

%Recall that our approach is based on probabilistically replacing terms by similar terms or synonyms.
%To achieve a high degree of freedom for our anonymization mechanism to choose
%similar words, but hopefully with different spelling to disturb typical authorship attribution features,
%we can try to increase the number of possible alternative candidates by expanding the extracted vocabulary
%by synonyms of the extracted terms.
%Therefore, we add an option to extend our base vocabulary by the known synonyms of each term
%in order to increase the number of well-suited candidates for choosing synonyms during our anonymization.
%Our implementation uses the synonyms as modeled in WordNet's synsets.

\subsubsection{Similarity Rating Function}
\label{sec:params:similarity-functions}

%There are different ways to define the rating function $\rating(v,w)$
We now describe the rating function $\rating(v,w)$ %our implementation of 
that expresses the suitability of a substitute term $w$ for an input term $v$.
%We will describe one fundamental technique, plus an optional tweak that can be applied on top.
%In this section, we will present some possible ideas to define a rating function $\rating$
%to assign the level of similarity between an input and output term.
%\paragraph{Word Vectors}
One fundamental technique are \emph{word vectors} or \emph{embeddings}
which are dense vector representations of words in a real vector space.
They are commonly derived with the intention that similar words have embeddings in the vector space that are nearby.
%thus allowing to measure similarity via the distance or angle between them.
We can therefore compute the similarity between two words
simply and efficiently as \emph{cosine similarity} %$$\cos(v,w):=\frac{\vec{v}^T \vec w}{\norm{\vec v}_2 \norm{\vec w}_2}$$
between their corresponding word vectors. %$\vec v$ and $\vec w$.
\TODO{660B-5: Address context-(in)senstitvity.}
Two recent models to derive word vectors that achieve high accuracy in word similarity and analogy benchmarks
are \enquote{word2vec} \cite{mikolov2013distributed, mikolov2013efficient}
and \enquote{GloVe} \cite{pennington2014glove}.
%These methods use neural networks to train word vectors in such a way that the angle between two word vectors is smaller
%if the words frequently appear in similar contexts.
%This is motivated by the distributional hypothesis by Harris \cite{harris1954distributional},
%which states that words that appear in similar contexts typically have similar meanings.
%Therefore, we can use the cosine similarity between two terms as a measure of similarity,
%and hence as rating function.
%The similarity between two words can be computed simply and efficiently
%as \emph{cosine similarity} %$$\cos(v,w):=\frac{\vec{v}^T \vec w}{\norm{\vec v}_2 \norm{\vec w}_2}$$
%between their corresponding word vectors. %$\vec v$ and $\vec w$.
%An advantage of using word vectors is that computing the cosine similarity is relatively fast and simple.
%On the downside, they only represent individual words and hence cannot distinguish ambiguous words
%(with several meanings) from their context.
%To address this issue, the idea has since been extended to \emph{paragraph vectors} \cite{le2014distributed}, 
%where embeddings are derived for paragraphs, smaller passages of text, or even entire documents
%instead of single words (cf. future work).

%\paragraph{Tweaks based on Spelling}
%We can combine or modify the rating function $\rating$ for our mechanism with other functions
%to influence the choice of which word will be used as substitute for the original input term.
%We propose a type of modification to the rating function with the goal of further affecting the attack:
As we saw in \cref{sec:prevent-aa},
features such as the frequency of certain words and character $n$-grams
often make an essential and decisive contribution to authorship attribution methods.
Suppose we can choose a substitute for a given input term from a set of candidates with comparable similarity rating.
Then to best prevent the attack, it is beneficial to pick the candidate that differs most in spelling
from the input in order to obscure our word and $n$-gram frequencies.
We can achieve this by including the (normalized) Levenshtein or $n$-gram distance
in the rating function for the terms.
Note that care must be taken to weight this appropriately --
a too strong preference for differently-spelled substitutes will often pick completely different words
%that are spelled completely differently
that also have a different meaning from the original word, thus also negatively affecting the utility.

%\paragraph{Rating Function Implementation}
We have implemented the word similarity rating function as %$\rating(v,w)$ as cosine similarity
\[
	\rating(v,w) := \cos(v,w) - s B(v,w),
\]
where $\cos(v,w)$ is the cosine similarity between the corresponding GloVe \cite{pennington2014glove} word vectors,
%$\vec v$ and $\vec w$ %for the input word $v$ and candidate word $w$.
and $B(v,w)\in[0,1]$ is the \emph{bigram overlap}, i.e. the proportion of matching letter bigrams in $v$ and $w$.
The scaling factor $s$ determines if and how strong the bigram overlap affects the rating.
%
% TODO OPTIMIZATION -> mote to where it applies (general algo, evaluation, ...)
As optimization, we precompute the word similarity ratings and probabilities for the Exponential mechanism
for the entire vocabulary, which yields a significant performance boost. %if the entire similarity matrix fits into main memory.

%\subsection{Evaluation of Utility and Attack}
\subsection{Experiment Description}
\label{sec:results:evaluation}

In this section, we describe the context and setup of our evaluation.
%We then continue with a presentation and discussion of the results.

%\subsubsection{Experiment Description}
% DATASET
\paragraph{Dataset}
We perform a series of experiments with our algorithm
on the ``20 newsgroups'' dataset\footnote{\url{http://qwone.com/~jason/20Newsgroups/}}.
It comprises almost 19,000 postings from 20 different newsgroups,
and comes with predefined train (60\%) and test (40\%) sets
which we use throughout our experiments.
For the text classification task, a label is provided for each message indicating the corresponding newsgroup.
For the authorship attribution task, we extracted the ``From`` field in the header of each message
and use it as author identifier.
%each corresponding to a certain topic.
%The analyst's task (utility) is to classify each unknown message to one of these 20 topics.
%Moreover, each posting contains meta-data in its header including the email address of its sender,
%which we extract and use as identity of the corresponding author for the attacker's authorship attribution task.
%This allows us to also perform authorship attribution (attack) on the messages,
%making the 20 newsgroups dataset particularly suitable for our evaluation.
%so the same dataset provides labels for utility \textit{and} attack
%and hence can be used for both tasks simultaneously.
Note that we strip header and footer data %meta-data including headers and footers
before performing the actual classification and identification tasks
as to make them more realistic.
%(The \textit{subject} field could provide hints on the topic,
%and the \textit{from} field directly reveals the identity of the sender.)
%This causes the baseline classification accuracy to drop from over 85\%
%to around 70\% with the meta-data removed.

% SCENARIOS
%Utility: Classify each posting to its newsgroup, attack: assign posting to its author (identified by email address).
\paragraph{Attack Scenarios}

After filtering out missing and ambiguous identifiers, we count 5735 authors, %in the train set,
but the majority provides insufficient training samples
%(below 5 for 5344 authors,
%(below 10 for 5613 authors, 
(below 20 for 5711 authors)
for properly fitting a model.
We therefore evaluate the attack only for the ``top'' authors with the largest number of messages in the dataset.
Since the number of candidate suspects from which the correct author has to be determined
also can influence the authorship attribution performance,
we evaluate the attack for the \emph{top 5} and \emph{top 10} authors.
\Cref{tab:attack-scenarios} provides the number of train and test messages per author.
% for each attack scenario.
%Numbers of train and test \emph{per author} are provided in \cref{tab:attack-scenarios}.

Another issue with the dataset %in conjunction with our evaluation
is that some users are active in only a single newsgroup,
in which case knowledge of authorship (attack) implies knowledge of the targeted newsgroup (utility).
%From a theoretical point of view, however, not all is lost:
%While this allows the malicious attacker to also perform the benign classification task,
%the implication still allows our algorithm to make the attack more difficult %than the classification task
%since there typically are many postings from several authors in each newsgroup.
%To analyze the effects of authors actively participating in one or multiple newsgroups,
%we distinguish two subsets of authors:
We therefore devise two subsets of authors:
%that allow us to analyze the effects when the two tasks are more or less similar:
\begin{description}
	\item[Any] Each suspect author can have postings in any number (one or more) of newsgroups.
	\item[Multi] Each author must be active in \textit{at least two} different newsgroups.
%	The idea is to reduce the similarity between the attacker's and analyst's tasks
%	to allow a clearer distinction when evaluating the impact of our anonymization technique.
\end{description}
The idea of the ``Multi'' group is to reduce the similarity between the attacker's and analyst's tasks
to allow a clearer distinction when evaluating the impact of our anonymization technique.
%where each author has posted in at least one or two different newsgroups,
%with the aim of reducing the similarity between the attacker's and analyst's tasks.

%Combining the cases for the number of active newsgroups and the number of top authors,
%We hence arrive at a total of four scenarios for which we evaluate the attack (cf. \cref{tab:attack-scenarios}).

\begin{table}
%	\small
	\caption{Attack scenarios with minimum \emph{per author} numbers
		for active groups and train/test messages in the dataset.}
	\label{tab:attack-scenarios}
	\begin{tabular}{ccccc}  
		\toprule
		%					&		   &		& \multicolumn{2}{c}{\# messages} \\
		Scenario    & Suspects & \#Groups & \#Train & \#Test \\
		\midrule
		Top 5/Any	& Top 5  & $\geq1$ 	& $\geq35$ & $\geq17$ \\
		Top 10/Any 	& Top 10 & $\geq1$ 	& $\geq28$ & $\geq9$ \\
		Top 5/Multi	& Top 5  & $\geq2$ 	& $\geq29$ & $\geq9$ \\
		Top 10/Multi& Top 10 & $\geq2$ 	& $\geq21$ & $\geq8$ \\
		\bottomrule
	\end{tabular}
\end{table}

\paragraph{Processing Pipeline.}

%\begin{figure}%[tbh]
%	\centering
%	\includegraphics[width=\columnwidth]{pipeline1.pdf}
%	\caption{The SynTF processing pipeline.}
%	\label{fig:pipeline}
%\end{figure}

\begin{figure*}
	\centering
	\begin{subfigure}{\columnwidth}
		\includegraphics[width=\textwidth]{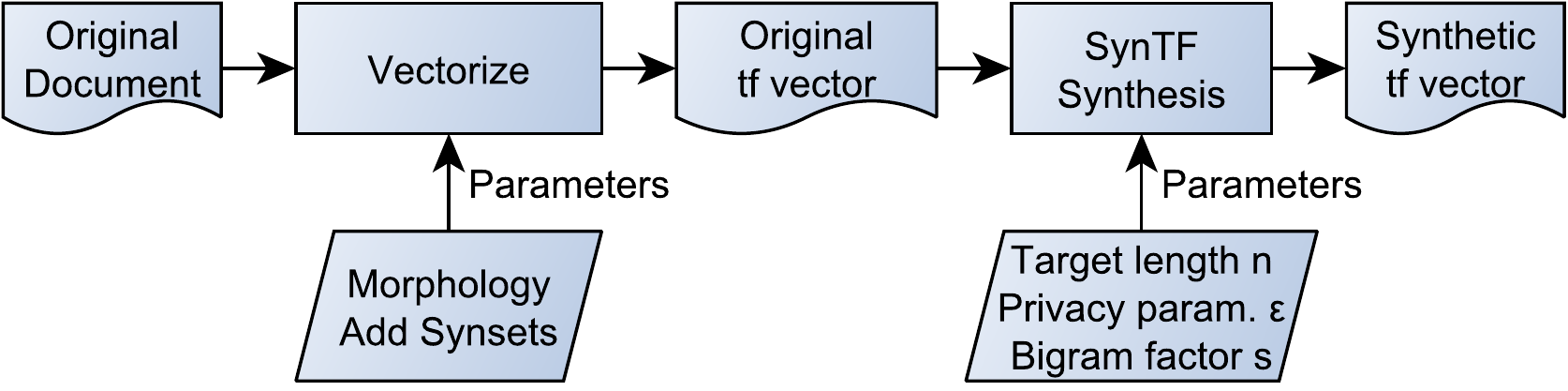}
		\caption{The main SynTF processing pipeline, including parameters.}
		\label{fig:pipeline}
	\end{subfigure}
	\hfill
	\begin{subfigure}{\columnwidth}
		\includegraphics[width=\textwidth]{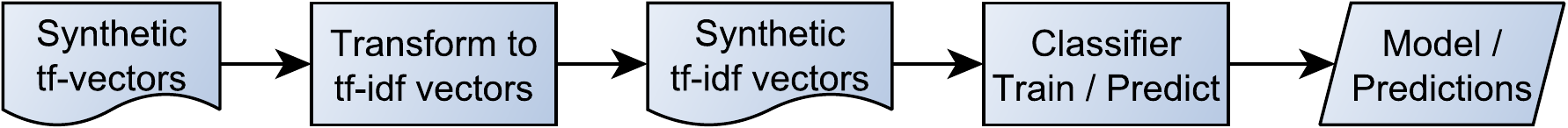}
		\caption{Subsequent text classification pipeline.}
		\label{fig:analyst}
		%	\end{subfigure}
		\vspace{1ex}
		%	\begin{subfigure}{0.95\columnwidth}
		\includegraphics[width=\textwidth]{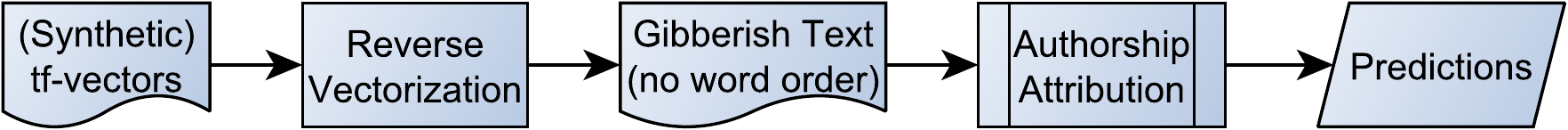}
		\caption{Subsequent authorship attribution pipeline.}
%			Reverse-vectorized texts are required to (at least partially) deduce WritePrints features.}
		\label{fig:attacker}
	\end{subfigure}
	\caption{Processing pipelines for the main SynTF mechanism and subsequent analyst and attacker tasks.} %
\end{figure*}

All documents traverse a processing pipeline that can be broken down into three parts:
For each document, the main \algname pipeline (\cref{fig:pipeline})
first produces a synthetic \gls{tfvec} (cf. \cref{sec:approach:dp-shuffle}).
It can be influenced by a number of parameters as described in \cref{sec:implementation}.
Next, the synthetic \glspl{tfvec} traverse the analyst's text classification pipeline (\cref{fig:analyst})
and the attacker's authorship attribution pipeline (\cref{fig:attacker})
to measure the prediction performance for each task.
In both cases, we evaluate a multinomial na\"ive Bayes classifier and a linear SVM.
%We use the train and test sets for training and testing respectively.
%
We perform 10 runs of the entire pipeline (anonymization + evaluation) for each combination of parameters
to reduce fluctuations and get stable results.

%The \algname mechanism (\cref{fig:pipeline}) anonymizes each document independently.
%It first vectorizes the input document to produce a term-frequency vector.
%Each original \gls{tfvec} is then converted to a synthetic \gls{tfvec}
%by the main \cref{alg:synth-tf}.
%Both vectorization and synthesis can be influenced by a number of parameters,
%as described in \cref{sec:implementation}.
%The entire process is depicted in \cref{fig:pipeline}.

%\paragraph{Analyst and Attacker Tasks}
The analyst (cf. \cref{fig:analyst}) first transforms the \glspl{tfvec} to \acrshortpl{tfidfvec}
which are commonly used in classification tasks.
He then trains a classifier with the \emph{training} subset of the dataset,
and subsequently uses it to predict the newsgroups for the \emph{test} subset.
%Both utility and attack are evaluated by first training a classifier
%using the \textit{training} subset of the 20 newsgroups dataset,
%and then having the trained classifier predict the newsgroup or the author
%on the \textit{test} subset of the dataset.
We implement the classification in Python based on \texttt{scikit-learn} \cite{scikit-learn},
using its \texttt{MultinomialNB} classifier with smoothing ($\alpha=0.01$),
and its \texttt{LinearSVC} classifier with default parameters ($C=1$).

For the attack, we make use of the ``JStylo'' authorship attribution framework \cite{mcdonald2012use}.
%\footnote{\url{http://psal.cs.drexel.edu/}}.
%(available at \url{http://psal.cs.drexel.edu/}).
%from Drexel University's Privacy, Security, and Automation Laboratory\footnote{\url{http://psal.cs.drexel.edu/}}.
%(cf. \cite{abbasi2008writeprints})
It supports several extended feature sets such as ``WritePrints'' proposed in \cite{abbasi2008writeprints}.
WritePrints includes additional stylistic features (cf.~\cref{sec:prevent-aa}) on top of the usual \gls{BoW}
that have to be extracted from full texts.
However, since the attacker only gets synthetic \glspl{tfvec} and not full texts,
she first converts the numbers in the \glspl{tfvec} to text by repeating each word accordingly,
which allows at least partial deduction of WritePrints features (``reverse vectorization'' in \cref{fig:attacker}).
Note that the ``full'' WritePrints feature set contains a virtually endless number of features
and severely degrades performance (speed).
Furthermore, the authors of \cite{mcdonald2012use} have shown that despite its title,
the ``limited'' version even outperforms the ``full'' WritePrints in terms of accuracy,
which we could confirm in own experiments.
Therefore, we keep the default JStylo configuration with the ``\acrlong{WPL}'' feature set.
\TODO{%660B-4: WritePrints full vs. limited. Full worse \cite{mcdonald2012use}.
	Include own results confirming degradation??}
JStylo builds on the Weka machine learning library.
We use its \texttt{NaiveBayesMultinomial} classifier with Laplace smoothing
and its \texttt{SMO} SVM classifier with linear kernel and $C=1$ by default.

%We evaluate both tasks' performance (measured as weighted F1-score) 
%for a multinomial na\"ive Bayes classifier and a linear SVM.
%
%We implement the newsgroup classification in Python based on \texttt{scikit-learn} \cite{scikit-learn},
%using its \texttt{MultinomialNB} classifier with smoothing ($\alpha=0.01$),
%and its \texttt{LinearSVC} classifier with default parameters ($C=1$).
%
%JStylo builds on Weka as underlying machine learning library and provides
%the \texttt{NaiveBayesMultinomial} and \texttt{SMO} SVM classifiers,
%the former with Laplace smoothing and the latter configured with a linear kernel and $C=1$ by default.

\paragraph{Finding Optimal Parameters}

\begin{table}%[htb]
%	\small
	\caption{Evaluated and \underline{optimal} SynTF parameters.} %(alg.~\ref{alg:synth-tf})
	\label{tab:eval-params}
	%	{\small The top half shows vectorization parameters,
	%		the bottom half shows synthesis/randomization parameters. Optimal parameters are \underline{underlined}.}
	\begin{tabularx}{\columnwidth}{llX}  
		\toprule
		%		\multicolumn{2}{c}{Item} \\
		%		\cmidrule(r){1-2}
		Parameter    & Values & Description \\
		\midrule
		morpho- & \underline{lemma} & Lemmatize words. \\
		logy    & lower             & Convert words to lower case. \\
		        & orth              & Leave spelling unchanged. \\
		\addlinespace[1pt]
		synsets & true/\underline{false} & Extend vocabulary with additional synonyms from WordNet. \\
		\midrule
		$s$ & 0, 0.1, 0.2, 			& Impact factor of letter bigram \\
		    & \underline{0.3}, 0.4	& overlap on rating function $\rating$. \\
		\addlinespace[1pt]
		$n$ & 100, \underline{150}, 200 & Length of output vector (words). \\
		\addlinespace[1pt]
		$\epsilon$ & 35--55 (\underline{47.5}),		& Privacy parameter (stepsize 2.5). \\
		           & effectively $\underline{25.4}$	& Effective loss $\ell$,
		           									  cf. sec.~\ref{sec:shuffle:tight-bounds}. \\
		\bottomrule
	\end{tabularx}
\end{table}

We perform a grid search over the \algname parameters %of \cref{alg:synth-tf}
listed in \cref{tab:eval-params} to find ``optimal'' parameters
in the sense that they should simultaneously strongly affect authorship attribution
but mostly leave classification into newsgroups unaffected.
As metric to find these optimal settings
we use the difference between the relative performance impacts on utility and attack:
%Formally, denote by $\alpha_a(p,s)$ and $\alpha_u(p,s)$ the F1-score of the attack and utility, respectively,
%where $p$ is a set of parameters as indicated in \cref{tab:eval-params},
%and $s\in\set{\text{plain},\text{vector},\text{synth}}$
%denotes the \textit{stage} of the data in the transformation process,
%that is whether the score was measured on the plain text,
%the vectorized feature vector, or the synthetic feature vector.
%The relative performance impact for attack and utility is defined as
%\[
%	\beta_x = \frac{ \alpha_x(p,\mathrm{synth}) }{ \alpha_x(p,\mathrm{plain})} - 1
%\]
%for $x\in\set{a,u}$.
Given parameters $\mathbf{p}$, denote by $\beta_U(\mathbf{p})$ the relative performance
of the analyst's classification task (measured as $F_1$ score),
% of the classifier trained to classify messages into their corresponding newsgroup)
and similarly denote by $\beta_A(\mathbf{p})$ the relative performance of the attacker's task.
Then the optimal parameters are $\mathbf{\hat p} = \mathrm{argmax}_{\mathbf{p}} (\beta_U(\mathbf{p})-\beta_A(\mathbf{p}))$.
Since we want them to equally cover all four attack scenarios,
we find optimal parameters that maximize the \emph{minimum} difference $\beta_U(\mathbf{p})-\beta_A(\mathbf{p})$
over all attack scenarios.
%of their relative utility and attack performances.
Furthermore, we perform 10 runs of the anonymization--evaluation process
for each combination of parameters to reduce fluctuations and get stable results.

\subsection{Discussion of Results}

%We ran the analyst's classification and attacker's authorship attribution tasks
%in parallel to benefit from modern multi-core CPUs.
After running the evaluation, we found the optimal parameters highlighted in \cref{tab:eval-params}
with privacy parameter $\epsilon=47.5$.
However, our tight bounds analysis (cf. \cref{sec:shuffle:tight-bounds}) 
shows that the effective privacy loss $\ell(\Em)\approx 25.4$ is only about half as large. %as $\epsilon$,
%thus yielding $25.4$\-/\dlp.
%lemma  False    2.0     0.3  150.0  47.5  25.384567
%morph=lemma, synsets=false, bigram scale 0.3, target length 150,
%and privacy parameter $\epsilon=47.5$ with the effective privacy loss 25.384567.
\Cref{tab:eval-results} provides exemplary performance figures
in the ``Top 10/Any'' scenario for both topic classification and authorship attribution.
\TODO{(660C-2) Point out Aggreg. vs. Indep. Analysis}
\Cref{fig:cmp-best} depicts the relative performance 
between utility (green lines, left y-axis) and attack (red lines, right y-axis)
in the different stages of \algname. %\cref{alg:synth-tf}.
The bottom x-axis indicates the privacy parameter $\epsilon$,
with the corresponding effective privacy loss values $\ell(\Em)$ on the top.
%(cf.~\cref{sec:shuffle:tight-bounds}).
The dotted, dashed, and solid lines mark the utility and attack performances
with the original (plaintext), vectorized, and synthetic data, respectively,
where we used the optimal parameters for vectorization and synthesis as mentioned above.
% att-util; org-syn; SynTF-scrub

\begin{figure*}
	\centering
	\begin{subfigure}[t]{0.498\textwidth}
		\centering
		\includegraphics[width=\textwidth]{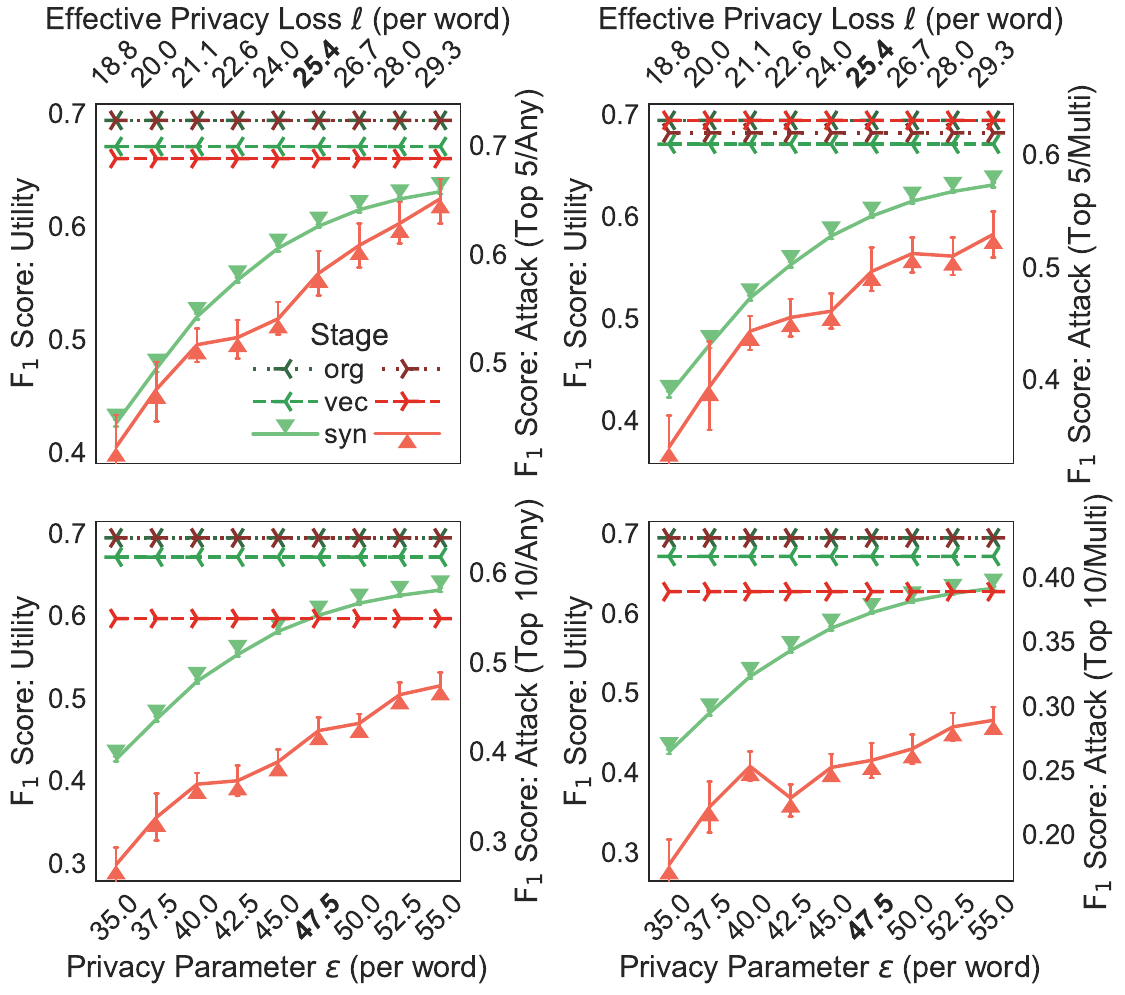}
		\begin{minipage}{0.96\textwidth}
			\caption{Multinomial na\"ive Bayes.}
				%				Impact on authorship attribution clearly increases with the number of suspects.
				%				Authors with messages in multiple groups are harder to identify.
		\end{minipage}
	\end{subfigure}
	%	\hfill
	\begin{subfigure}[t]{0.498\textwidth}
		\centering
		\includegraphics[width=\textwidth]{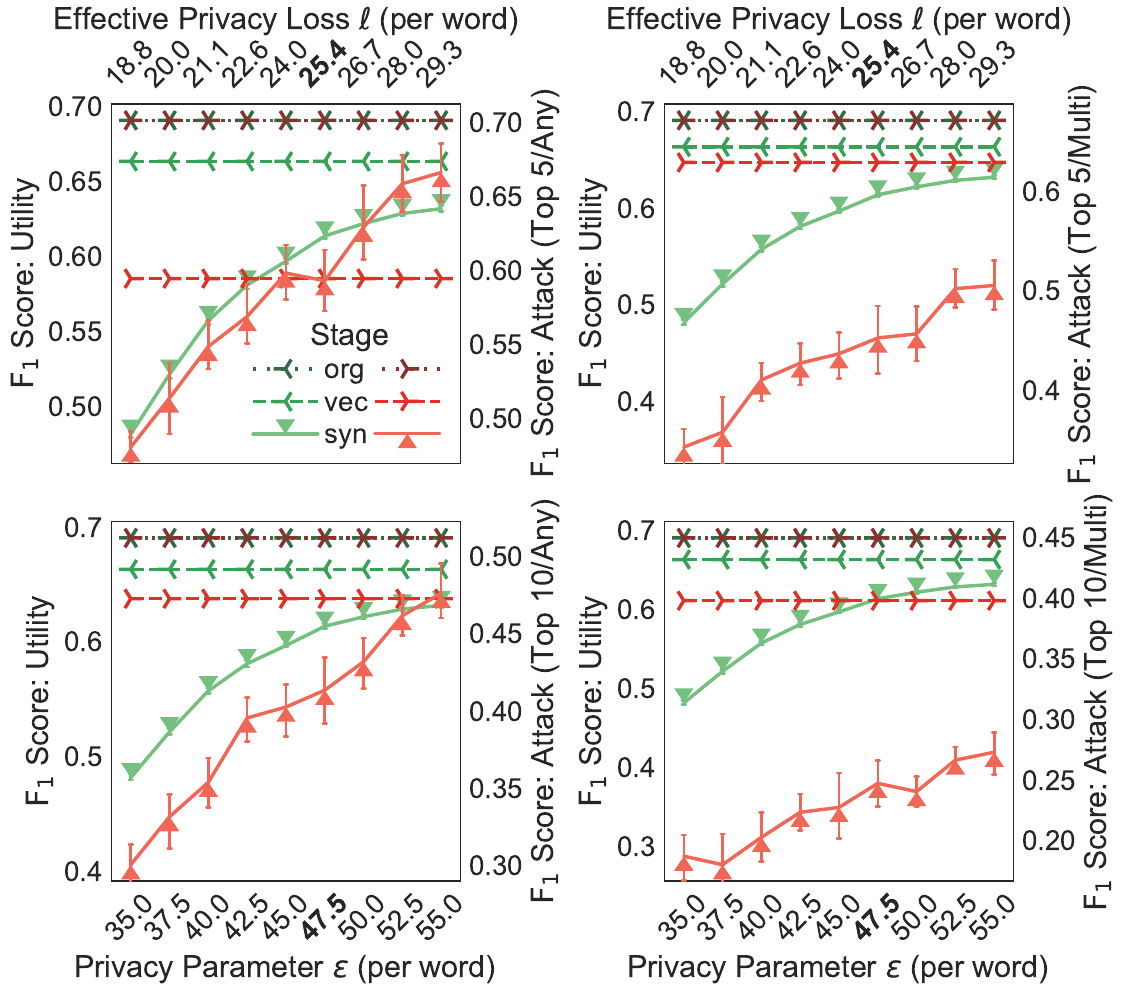}
		\begin{minipage}{0.96\textwidth}
			\caption{Linear SVM.}
				%				Impact on attack increases with number of authors and active groups.
				%				The curves for attack and utility performance clearly resemble each other
				%				in the ``Top 5/Any'' scenario.
				%Attack performance close to utility if authors who are active in only one group are included.
		\end{minipage}
	\end{subfigure}
	\caption{Relative performance of analyst (green) and attacker (red) tasks in different stages of the \algname process,
		per attack scenario (dotted: original data, dashed: \glspl{tfvec}, solid: synthetic \glspl{tfvec}).
		A (positive) gap between the green and red lines shows how much the attack is more affected than utility.
		Impact on attack increases with number of authors and active groups.%
	}
	%Attack scenarios differ in the set of suspects (top 5 or 10 authors)
	%and their number of active groups (any or at least 2).}
	\label{fig:cmp-best}
\end{figure*}

%We observe that the vectorization has a noticeably larger impact on the attack caused by
We observe that the vectorization already affects the attack more due to
the loss of structural and syntactic features, except in one case (Top 5/Multi).
%The solid line marks the performances of both tasks with the synthetic data.
Note that the size of the (positive) gap between the green and red lines
indicate the analyst's gain over the attacker in terms of the relative performance
of the corresponding stage of the anonymization.
Obviously both utility and attack suffer with a decreasing privacy parameter $\epsilon$.
However, in most cases the gap between analyst and attacker is even higher than after vectorization,
which indicates a growing advantage for the analyst.
Furthermore, it shows that our \algname mechanism successfully impairs authorship attribution
while having only a mild effect on the classification task.% with the optimal parameters.

\paragraph*{Impact of Attack Scenarios}
Comparing the four scenarios with respect to the gap size,
we make the following deductions:
%\todo{save space (no description env.)}
%\begin{description}
%	\item[Number of suspects]
As expected, authorship attribution quickly becomes harder
with an increasing number of suspect authors.
%	\item[Group diversity] %diversity impact -- resemblance attack/utility when considering authors in one group
Similarly, excluding authors who are active in only one newsgroup
widens the gap, as we can see when going from the ``Any'' to the ``Multi'' scenarios.
This indicates that	our method is even more effective
when the benign and malicious tasks are actually based on \emph{distinct} problems.
%\end{description}

\paragraph{Impact of Parameters from \Cref{tab:eval-params}}
%We now discuss the impact of individual parameters .

A key factor in the success of our method is the letter bigram overlap $B$
in the rating function $\rating$.
%Without its effect of preferring good synonyms with differing spelling,
%the anonymization method would not work as well in preventing authorship attribution attacks.
Its effect of preferring synonyms with different spelling
improves the capability of our method to prevent authorship attribution attacks.
%
%We illustrate the impact of the bigram overlap factors $s$ in \cref{fig:bigram-impact}:
We illustrate this effect depending on the bigram overlap factor $s$ in \cref{fig:bigram-impact}:
Without bigram overlap ($s=0$), the attacker has an advantage in all ``Top 5'' scenarios (red bars). 
Only when $s \geq 0.3$, we see a shift of power in favor of the analyst (green bars).
In the ``Top 10'' scenarios, the analyst enjoys an advantage even without the bigram overlap,
but we can roughly double his advantage if we choose the optimal value $s=0.3$.

\begin{figure}
	\centering
	\includegraphics[width=0.96\columnwidth]{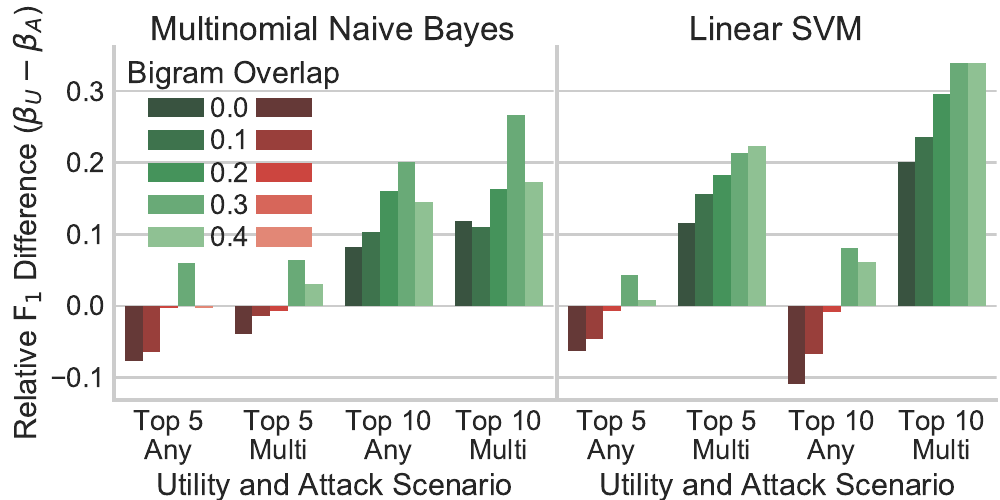}
	\caption{Impact of letter bigram overlap factor $s$.} %, by classifier type.
%		Positive values (green) show better relative gains for the analyst.}
%		Positive values (green) indicate a more detrimental effect on attack than on utility.}
%		in \algname on utility and attack performance, by classifier.}
	\label{fig:bigram-impact}
\end{figure}

Regarding morphology,
observe that the use of upper and lower case letters is a stylistic feature
that can pose a clue for authorship attribution but barely has any relevance for topic inference.
%is a stylistic feature and thus poses a clue for authorship attribution.
Therefore, transforming all words to lowercase affects the attacker more than the analyst.
%but mostly preserve the analyst's ability to classify the messages into topics.
Lemmatization strips off word endings and hence reduces the attacker's
information on writing style further, but it also has an impact
on classification since the meaning can change between a word and its lemma.
Still, in terms of our definition of ``optimal'' parameters,
using lemmatized words gave the best relative performance gain for the analyst,
indicating that the lost word endings are more severe for the attack.

Other parameters are less insightful: %are subtle and less profound:
Increasing the output length will help increase both tasks' performance,
however, the gain becomes less for larger output lengths.
Moreover, the inclusion of additional synonyms in the vocabulary
did not provide any benefit.

\paragraph{SVM Anomaly}
We observe one anomaly in the ``Top 5/Any'' scenario for the SVM.
Apparently, vectorization already causes a drastic reduction of the attack performance.
However, for $\epsilon\geq45$, going from vectorized to synthetic vectors
\emph{increases} the attack performance.
This is unexpected since the information \emph{lost in vectorization}
will not be restored by the synthesis process.
Our current hypothesis is that the SVM might overfit on the vectorized training data,
causing poor predictions on the vectorized test data,
and the randomness in the synthesis step in turn acts as regularization.

%\subsubsection{Amazon product data?}
%Huge dataset ($>41$ mio reviews), need to find core subset of reviews for a small number of products/product categories and users with reviews for many of these (similar issue as with 20 newsgroups: reviewer has preference for products from some category, so reviewer ID often implies product category).

\begin{table}
%	\small
	\caption{Evaluation results (Top 10/Any)}
	%	{\small The top half shows vectorization parameters,
	%		the bottom half shows synthesis/randomization parameters. Optimal parameters are \underline{underlined}.}
	\label{tab:eval-results}
	\begin{tabular}{rllllllll}
		\toprule
		& \multicolumn{3}{c}{Utility} & \multicolumn{3}{c}{Attack} & Gain \\ %& \multicolumn{2}{c}{Impact} \\
		\cmidrule(r){2-4} \cmidrule(lr){5-7} \cmidrule(lr){8-8}
		Method			& $F_1$ & P & R & $F_1$ & P & R & $\Delta F_1$ \\ %& $\beta F_1$ \\
		\midrule
		none (original) & 0.69 & 0.71 & 0.70 & 0.64 & 0.71 & 0.63  & 0.06 \\
		\addlinespace[1pt]
		\algname abs.	& 0.60 & 0.61 & 0.61 & 0.42 & 0.44 & 0.43  & 0.18 \\
		scrubadub abs.	& 0.64 & 0.65 & 0.65 & 0.57 & 0.63 & 0.57  & 0.06 \\
		\addlinespace[1pt]
		\algname rel.	& 87\% & 86\% & 87\% & 66\% & 61\% & 69\%  & 20\% \\
		scrubadub rel.	& 92\% & 92\% & 92\% & 90\% & 88\% & 91\%  & 02\% \\
		\bottomrule
	\end{tabular}
\end{table}

\subsubsection{Comparison with Scrubbing Methods}
\label{sec:eval:scrub}

%\subsubsection{scrubadub}
We run the open source \texttt{scrubadub} (\url{http://scrubadub.readthedocs.org/})
%\footnote{\url{http://scrubadub.readthedocs.org/}}
tool on the 20 newsgroups dataset to remove \gls{PII}
and evaluate the utility and attack performance in our scenarios.
\Cref{fig:cmp-scrub} shows a comparison of the results with our \algname method
and optimal parameters.
The results indicate that our method outperforms the scrubbing technique
in preventing the attack in all four attack scenarios,
at a comparable level of utility.
For instance, in the ``Top 10/Any'' scenario listed in \cref{tab:eval-results},
\algname achieves an $F_1$ score of $0.60$ for classification,
where scrubadub is slightly better with $0.64$, down from $0.69$.
For the attack, however, scrubadub drops from $0.64$ to $0.57$, %which is still quite large,
whereas \algname manages to more than triple the reduction
and push the attacker's performance down to $0.42$.

\begin{figure}
	\centering
	\includegraphics[width=0.93\columnwidth]{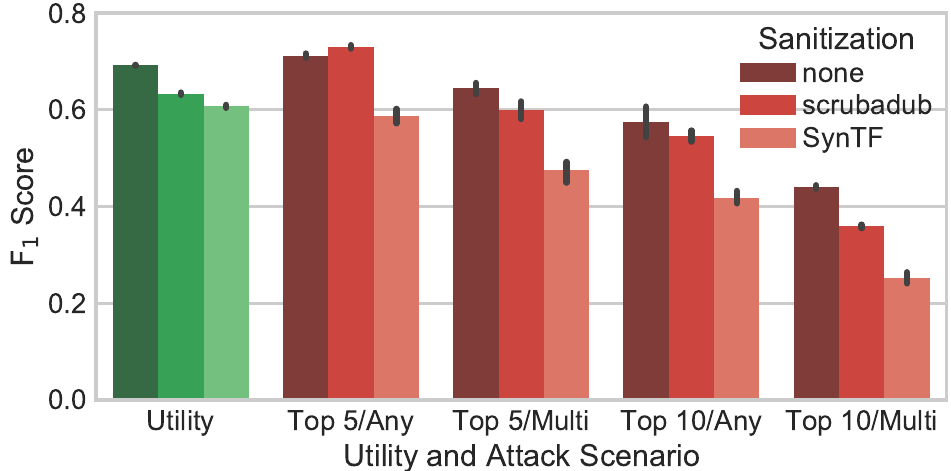}
%	\caption{Comparison between traditional de-identification (scrubbing) and our \algname method.}
	\caption{Comparing \algname and traditional data removal.}% ``scrubbing''.}
	\label{fig:cmp-scrub}%\subsubsection{Na\"ive Approach?}
	%	\end{subfigure}%tdb.
\end{figure}%\TODO{Decide if worthwhile.}

\section{Related Work} % TODO and Discussion?
\label{sec:related}

%In this section, we discuss other approaches to anonymize textual data and prevent authorship attribution.
\TODO{Reduce to one page. Add references for large-$\epsilon$ \dlp with individual analysis/local model.}

\paragraph{Authorship Obfuscation}
% From Introduction (following discussion on authorship attribution techniques)
Several countermeasures against authorship attribution have been proposed.
Rao and Rohatgi \cite{rao2000can} examine newsgroups postings and identify the authors
from the body of the text by analyzing the frequency of function words.
%which are words such as articles or pronouns with little or no lexical
%but grammatical meaning that expresses relationships between parts of a sentence.
They suggest to either use automated \emph{machine translation} to a foreign language and back,
or to \emph{educate authors} who want to write anonymous documents about authorship attribution attacks.
However, these countermeasures are insufficient:
In \cite{caliskan2012translate}, Caliskan et al. show that authorship attribution is still possible
even after performing multiple machine translations. %translated multiple times using machine translation services.
Furthermore, Afroz et al. \cite{afroz2012detecting} show that \emph{deceptive writing}
by an author trying to imitate another or to obfuscate his own writing style
can be detected with high accuracy.

Anonymouth \cite{mcdonald2012use} is based on JStylo
and uses clustering of two references sets with the author's and foreign sample texts
to propose manual changes that have to be made to the document.
% would prevent authorship attribution. %of the modified document.
%The author has to manually edit the document to implement the changes
The process must be repeated until authorship attribution is prevented sufficiently.
%A successor named Worden with improved usability has since been published\footnote{\url{http://www.worden.info/}}.
%
Kacmarcik and Gamon \cite{kacmarcik2006obfuscating} follow a similar
but more automated approach based on decision trees and SVMs.
Their method adjusts the \gls{tfvec} of a document by moving
its feature values closer to those of other writers,
as to prevent the classifier from identifying the correct author.
%
%to more closely resemble the values of other authors.
%Note that instead of making adjustments to the original text, they change the feature vectors directly.
While the countermeasure is effective against the evaluated SVMs with up to 70 features,
the more sophisticated \emph{unmasking} approach by Koppel and Schler \cite{koppel2004authorship, koppel2007measuring}
is still able to distinguish the actual author from others.
Kacmarcik and Gamon in turn propose a ``deep obfuscation'' variant of their method
which iteratively tries to make unmasking harder, however,
this requires more and more changes to be made to the documents.
%Unfortunately, they did not evaluate the suitability of the anonymized feature vectors for text mining purposes.
%
%If an attacker knows the writing style of the original and target authors, he could still identify...?
%\TODO{Check: Unmasking, unlike the SVMs, is trained with modified documents?
%	Then we can argue that our method, althounh non-adaptive, also prevents attack on modified data!}
%\TODO{Also mention that they (unmasking) train on original data (non-adaptive adversary),
%	while we allow the adversary to build a model on the anonymized data (adaptive adversary)
%	which improves accuracy. Include comparison with our method in evaluation (org-org, org-syn, syn-syn)}

The results indicate that both methods are successful in preventing authorship attribution attacks in theory.
However, the authors of Anonymouth \cite{mcdonald2012use} observed
that while users were able to implement the suggested changes for very small feature sets with only 9 features,
they were overstrained with the changes for the more complex ``WritePrints (Limited)'' feature set 
which we also used in our experiments.
Similarly, Kacmarcik and Gamon \cite{kacmarcik2006obfuscating} observed that for a deep level of obfuscation,
one would have to consider more and more features and make corresponding changes to the document,
thus increasing the complexity for the user.
In practice, both methods seem cumbersome for the user if a deep level of obfuscation shall be reached.
Furthermore, both methods only prevent authorship attribution
with respect to a \emph{specific} reference corpus with other authors.
While our method does not produce human-readable texts,
it requires no manual changes to the documents,
and its protection is independent of a reference corpus.

\paragraph{De-Identification}
De-identification (or \emph{scrubbing}) methods provide a way
to remove \acrfull{PII} from textual documents.
They are often motivated by the health care and medical sectors and
focus on identifying and removing particular types of personal information
such as protected health information (PHI),
a list of 18 identifiers %\footnote{\url{http://www.ucdmc.ucdavis.edu/compliance/guidance/privacy/deident.html}}
as specified in the US Health Insurance Portability and Accountability Act (HIPAA) \cite{hipaa}. 
Popular methods include the ``Scrub System'' \cite{sweeney1996replacing},
the ``MITRE Identification Scrubber Toolkit'' (MIST) %\footnote{\url{http://mist-deid.sourceforge.net/}}
\cite{aberdeen2010mitre},
or the PhysioNet ``deid'' software package %\footnote{\url{https://www.physionet.org/physiotools/deid/}}
\cite{neamatullah2008automated}.
They typically work with lists of names and identifiers, regular expressions,
simple heuristics, and also machine learning techniques
to identify and remove pieces of text that constitute \gls{PII}.
%like ``scrubadub'' \footnote{\url{http://scrubadub.readthedocs.org/}}

While this kind of information must be removed
to protect the privacy of the subjects mentioned in the document,
our experiments in \cref{sec:eval:scrub} show that de-identification based on scrubbing
provides no adequate protection for the document's author although this is often critical,
as in the case of complaint letters or patient records
to protect the privacy of the treating physician. %in order 
Moreover, we found that publications on these methods
typically only evaluate their methods' ability to identify and remove
all pieces of \gls{PII} in the text
(cf. the survey by Uzuner et al. \cite{uzuner2007evaluating}).
We have not seen any evaluation on the impact of scrubbing %the de-identification
on further processing with text mining techniques such as document classification,
and more importantly, we have not found an evaluation whether and to what extend
these methods prevent authorship attribution techniques.

%
%\paragraph{Text Anonymization}
%Charu Agrawal Sketch-based
%
%\paragraph{Linkage Attacks?}
%E.g. Netflix, Smart Meters, AOL?
%We do not focus on this attack in this paper.
%(Consider text retrieval/ranking: Compute similarity of query and documents; usually via cosine similarity of query and document vectors. But then utility (inner product/cosine similarity) is almost the same as attack (correlation)! Hence either both will work well, or both will work poorly depending on the level of randomization. (obvious) conclusion: Hard to prevent attacks that are similar to utility!)
%
%\TODO{Argument for ignoring/neglecting this attack: e.g. if linkage attack says two documents are similar, they are not necessary by the same author, but more likely about the same topic. This should not be a privacy issue.
%	%What about e.g. patient records that could refer to the same patient? Similar: more likely that two records are identified as similar if they describe same disease, but not same person. (risky statement without quote.)
%}

%\paragraph*{Information Retrieval and NLP (Utility)}
%\begin{itemize}
%	\item text classification: e.g. news, spam
%	\item text search and retrieval (Salton, vector space model)
%	\item (basic sentiment analysis)
%\end{itemize}

\paragraph{Differential Privacy}
Differential privacy has been successfully applied to a wide range of problems
from simple statistical functions to machine learning.
The survey by Dwork \cite{dwork2008differential} provides a good overview of some earlier results.
It is commonly used to provide \emph{aggregate} statistics,
that is, multiple records are combined into one result.
A good example is RAPPOR \cite{erlingsson2014rappor},
which allows the collection of anonymized user statistics even over time.
However, releasing aggregate information only allows inferences on an entire population,
whereas we want to classify each document individually.
%hence our method must anonymize each record independently.
%
Releasing \emph{individual} data with an $\epsilon$ comparable to aggregating mechanisms
causes too much noise for individual records as it masks any difference
(topic, sentiment, etc.) between two inputs and hence prevents any utility.
The issue is well-known in the literature and has been observed
e.g. in the context of locations \cite{machanavajjhala2008privacy,andres2013geo},
graphs \cite{sala2011sharing}, and recommender systems \cite{machanavajjhala2011personalized}.
Approaches typically involve relaxing the privacy- or adjacency-definition 
\cite{chatzikokolakis2013broadening,andres2013geo,he2014blowfish}.
%(e.g. edge instead of node privacy for graphs).
%For instance,
Andrés et al. \cite{andres2013geo} circumvent the issue for location data
by generalizing \dlp to metrics \cite{chatzikokolakis2013broadening}.
%and choosing an appropriate radius of protection.
For graphs, Hay et al. \cite{hay2009accurate} define two varians of \dlp,
namely \emph{node} and \emph{edge} privacy, where two graphs are considered adjacent
if they differ either in an entire node (including its edges) or in just a single edge.
According to Kasiviswanathan et al. \cite{kasiviswanathan2013analyzing},
most works focus on the strictly weaker edge privacy %including \cite{sala2011sharing}
since it is harder to create node private algorithms providing good utility
with a comparable privacy loss.
For instance, Sala et al. \cite{sala2011sharing} revert to edge privacy for sharing graphs
and obtain usable results with $\epsilon=100$ \emph{per edge} (instead of \emph{per node}).
In comparison, our \algname mechanism achieves a privacy loss of only 25.4
\emph{per word} in the output (instead of \emph{per document}).

\section{Conclusion}
\label{sec:conclusion}

%In this paper, we 
We have presented \algname, a novel approach to produce anonymized,
synthetic \acrlongpl{tfvec} which can be used in lieu of the original \acrlongpl{tfvec}
in typical applications based on the \acrlong{VSM}.
Our method produces sparse vectors which are favorable regarding performance and memory efficiency.
%and the synthetic data is in the same format as the original data
%which potentially allows various text mining tasks and analyses to be performed.
%Benefits of our method include
%\begin{itemize}
%	\item 	
%	\item After determining the initial parameters,
%		no manual document changes are required to perform the anonymization.
%	\item Moreover, there is no dependence on corpus from authors to make features similar to.
%\end{itemize}
We have proved that our method fulfills \dlp
which currently serves as a ``gold standard'' for privacy definitions.
Since our method anonymizes each text individually,
it can be used locally at the data source to anonymize documents on-premise before collection,
e.g., to obtain anonymized training data for machine learning
or provide personalized ads based on anonymized emails or search queries. %(aggregated)
%Email providers such as Gmail could provide personalized ads
%based on the anonymized feature vectors instead of the original email.
%For instance, online marketing providers could use anonymized feature vectors
%of email messages or (aggregted) search queries
%to provide personalized ads while protecting their sources' identities.

Although our method requires a large $\epsilon$ to get reasonable utility,
we provide evidence that this is necessary:
%Firstly, we anonymize texts individually and analyze each record independently from the others,
%so it is desirable that an individual's record actually has an influence on the result.
First, we want to be able to analyze records independently from each other,
thus the anonymization must \emph{not} hide the influence of individual records in the result.
%Firstly, most other \dyp mechanisms are designed for aggregate statistics,
%where the designated goal is to hide the impact of an individual's record on the analysis result, 
%and where the error caused by the noise pertaining to each individual likely can be corrected or cancels out.
%In contrast, we are doing individual statistics
%where we actually do want 
Second, we have derived a necessary condition %and sufficient conditions
on the privacy parameter $\epsilon$ for the Exponential mechanism
%and a heuristic argument which states that the privacy parameter
indicating that it
must grow logarithmically in the size of the output space when high utility is required
but only a limited number of ``good'' outputs is available.
%given a limited number of ``good'' outputs.
To further address the issue, %with the large privacy parameter, 
we have derived alternative bounds on the privacy loss of the Exponential mechanism,
which in our case provide a substantial reduction of almost 50\%.

We have performed an extensive evaluation of \algname on the 20 newsgroups dataset
and analyzed the influence of different parameters.
Our results indicate that
%our method achieves our goal of affecting the attack performance
%more strongly than the utility performance.
%We have shown that 
it effectively prevents authorship attribution
%with a much larger impact than on classification (utility).
while facilitating tasks such as classification (utility).
In contrast, our experiments show that traditional scrubbing methods
are insufficient at preventing authorship attribution attacks.

\begin{acks}
The authors would like to thank the anonymous referees for their
valuable comments and helpful suggestions.
This work was supported by the European Union’s
Horizon 2020 Research and Innovation Programme
under grant agreement No.~653497 (PANORAMIX).
\end{acks}

\bibliographystyle{ACM-Reference-Format}
\bibliography{ms}

\end{document}